\documentclass[letterpaper,prl,superscriptaddress,twocolumn,10pt]{revtex4}
\usepackage{times,color,amsmath,graphicx}
\usepackage[colorlinks,breaklinks]{hyperref}
\hypersetup{breaklinks=true}
\usepackage[hyphenbreaks]{breakurl}

\usepackage{ulem}
\newcommand{\oper}[1]{\hat{#1}}
\newcommand{\ts}{{\Delta \tau}}
\newcommand{\idx}{m}

\usepackage{times}
\usepackage{amsmath,amsfonts}
\usepackage{color}

\newcommand{\REVISION}[1]{#1}
\newcommand{\operSB}[2]{\hat{#1}[#2]}
\newcommand{\SECONDREVISION}[1]{#1}

\usepackage{tikz}
\usetikzlibrary{matrix,patterns,shapes.arrows,arrows.meta,shapes}

\usepackage{pgfplots,pgfplotstable,amsmath,amsfonts,amsthm}
\usetikzlibrary{patterns,calc,calc,patterns,angles,quotes}

\usepackage{pgf}
\usepackage{tikz}
\usetikzlibrary{arrows,automata}
\usepackage[latin1]{inputenc}
\usepackage{verbatim}

\newtheorem{theorem}{Theorem}

\usepackage{times,color,amsmath,graphicx}
\usepackage[colorlinks,breaklinks]{hyperref}
\hypersetup{breaklinks=true}
\usepackage[hyphenbreaks]{breakurl}

\newcommand{\dt}{{\Delta\tau}}

\begin{document}

\title{\SECONDREVISION{Determining eigenstates and thermal states on a quantum computer \\ using quantum imaginary time evolution}}
\author{Mario Motta}
\email[Corresponding author. ORCID  0000-0003-1647-9864. E-mail: ]{mariomotta31416@gmail.com}
\affiliation{Division of Chemistry and Chemical Engineering, California Institute of Technology, Pasadena, CA 91125, USA}
\author{Chong Sun}
\affiliation{Division of Chemistry and Chemical Engineering, California Institute of Technology, Pasadena, CA 91125, USA}
\author{Adrian T. K. Tan}
\affiliation{Division of Engineering and Applied Science, California Institute of Technology, Pasadena, CA 91125, USA}
\author{Matthew J. O'Rourke}
\affiliation{Division of Chemistry and Chemical Engineering, California Institute of Technology, Pasadena, CA 91125, USA}
\author{Erika Ye}
\affiliation{Division of Engineering and Applied Science, California Institute of Technology, Pasadena, CA 91125, USA}
\author{Austin J. Minnich}
\affiliation{Division of Engineering and Applied Science, California Institute of Technology, Pasadena, CA 91125, USA}
\author{Fernando G. S. L. Brand\~{a}o}
\affiliation{Institute for Quantum Information and Matter, California Institute of Technology, Pasadena, CA 91125, USA}
\author{Garnet Kin-Lic Chan}
\email[Corresponding author.  ORCID  0000-0001-8009-6038. E-mail: ]{gkc1000@gmail.com}
\affiliation{Division of Chemistry and Chemical Engineering, California Institute of Technology, Pasadena, CA 91125, USA}

\begin{abstract}
  \SECONDREVISION{The accurate computation of Hamiltonian ground, excited, and thermal states on quantum computers
    stands to impact many problems in the physical and computer sciences, from quantum simulation to machine learning.
    Given the challenges posed in constructing large-scale quantum computers, these tasks should be carried out
    in a resource-efficient way. In this regard, existing techniques based on phase estimation or variational algorithms
    display potential disadvantages; phase estimation requires deep circuits with ancillae, that are hard to execute reliably without
    error correction, while variational algorithms, while flexible with respect to circuit depth, entail additional high-dimensional
    classical optimization.
    Here, we introduce the quantum imaginary time evolution and quantum Lanczos algorithms, which are analogues of classical algorithms for finding
    ground and excited states. Compared to their classical counterparts, they require exponentially less space and time per iteration,
    and can be implemented
    without deep circuits and ancillae, or high-dimensional optimization. We furthermore discuss quantum imaginary time evolution as a
    subroutine to generate Gibbs averages through an analog of minimally entangled typical thermal states. Finally,
    we demonstrate the potential of these algorithms via an implementation using exact classical emulation  as well as
    through prototype circuits on the Rigetti quantum virtual machine and Aspen-1 quantum processing unit.}
\end{abstract}
\maketitle

An important application for a quantum computer is to compute the ground-state $\Psi$ of a Hamiltonian
$\oper{H}$
\cite{Feynman_IJTP_1982,Abrams_PRL_1997}. 
This arises in simulations, for example, of the electronic structure of molecules and materials, 
\cite{Lloyd_Science_1996,Aspuru_Science_2005,kandala_2017_nature,kandala_2018_arxiv}
as well as in more general optimization problems. 
While efficient ground-state determination
cannot be guaranteed for all Hamiltonians, as this is a QMA-hard
problem 
\cite{Kempe_SIAM_2004}, several heuristic quantum algorithms 
have been proposed, including adiabatic state preparation with quantum phase estimation 
\cite{Farhi_MIT_2000,Kitaev_arxiv_1995} (QPE) and quantum-classical variational algorithms,
such as the quantum approximate optimization algorithm 
\cite{Farhi_MIT_2014,Otterbach_arxiv_2017,Moll_QST_2018} and variational quantum eigensolver 
\cite{Peruzzo_Nature_2013,McClean_NJP_2016,grimsley2018adapt}.
Despite many advances, these algorithms also have potential disadvantages, 
especially in the context of near-term quantum computing architectures with limited quantum resources. 
For example, phase estimation produces a nearly exact eigenstate, but appears impractical without error correction, while
variational algorithms, though somewhat robust to coherent errors, are limited in accuracy by a fixed Ansatz,
and involve high-dimensional noisy classical optimizations~\cite{mcclean2018barren}.

In classical simulations, different strategies are employed to numerically determine nearly exact ground-states.
One popular approach is imaginary-time evolution, 
which expresses the ground-state as the long-time limit of the imaginary-time Schr\"odinger equation 
$- \partial_\beta |\Phi(\beta)\rangle 
= \oper{H} |\Phi(\beta)\rangle$, $|\Psi\rangle 
= \lim_{\beta \to \infty} \frac{|\Phi(\beta)\rangle}{ \| \Phi(\beta) \|}$ (for $\langle \Phi(0) | \Psi \rangle \neq 0$). 
Unlike variational algorithms with a fixed Ansatz, imaginary-time evolution always converges to the 
ground-state, as distinguished from imaginary-time Ansatz optimization \cite{McArdle_arxiv_2018}.
\SECONDREVISION{Another family of approaches are variants of the iterative Lanczos method 
\cite{Lanczos_somewhere_1950}.}
The Lanczos iteration constructs the Hamiltonian matrix $\mathbf{H}$ in a Krylov subspace
$\{ |\Phi\rangle, \oper{H} |\Phi\rangle, \oper{H}^2|\Phi\rangle \ldots \}$; diagonalizing $\mathbf{H}$
yields a variational estimate of the ground-state which tends to $|\Psi\rangle$ for a large number
of iterations. For an $N$-qubit Hamiltonian, the classical complexity of imaginary time evolution and 
Lanczos algorithm scales as $\sim \exp{(\mathcal{O}(N))}$ in space and time.
Exponential space comes from storing $\Phi(\beta)$ or the Lanczos vector, while exponential time
comes from the cost of Hamiltonian multiplication $\oper{H} |\Phi\rangle$, as well as, in principle,
though not in practice, the $N$-dependence of the number of propagation steps or Lanczos iterations.
Thus it is natural to consider quantum versions of these algorithms that can overcome the exponential bottlenecks.

\SECONDREVISION{
Here we describe the quantum imaginary time evolution (QITE), the quantum Lanczos (QLanczos)
and the quantum analog of the minimally entangled typical thermal states (QMETTS) algorithm,
to determine ground-states, ground and excited states and thermal states on a quantum computer.}
\REVISION{Under the assumption of finite correlation length, these methods rigorously use exponentially reduced
space and time per propagation step or iteration, compared to their direct classical counterparts.
Even when such assumptions do not hold, the inexact versions of the QITE and QLanczos algorithms remain valid heuristics
that can be applied within a limited computational budget, and
offer advantages over existing ground-state quantum algorithms, as they do not
use deep circuits and converge to their solutions without non-linear optimization.}
A crucial common component is the efficient implementation of the non-Hermitian operation of an imaginary-time 
step $e^{-\ts \oper{H} }$ (for small $\ts$) assuming a finite correlation length in the state.
Non-Hermitian operations are not natural on a quantum computer and are usually achieved using ancillae and postselection,
but we describe how to implement imaginary time evolution on a given state without these resources.
The lack of ancillae and complex circuits make \SECONDREVISION{our algorithms} potentially suitable for near-term quantum architectures. 
We demonstrate the algorithms on spin and fermionic Hamiltonians 
using exact classical emulation, and demonstrate
proof-of-concept implementations on the Rigetti quantum virtual machine (QVM) and Aspen-1 quantum processing units (QPUs).

\noindent{\bf Quantum Imaginary Time Evolution}. Define a geometric $k$-local Hamiltonian $\oper{H} = \sum_\idx \operSB{h}{\idx}$ 
(where each term $\operSB{h}{\idx}$ acts on at most $k$ neighboring qubits on an underlying graph)
and a Trotter decomposition of the corresponding imaginary-time evolution,
\begin{align}
  e^{-\beta \oper{H}} = (e^{-\ts \operSB{h}{1}} e^{-\ts \operSB{h}{2}} \ldots)^n + \mathcal{O}\left( {\ts} \right); \ n= \frac{\beta}{\ts}
\end{align}
applied to a state $|\Psi\rangle$. After a single Trotter step, we have
\begin{align}
  |\Psi^\prime \rangle = e^{-\ts \operSB{h}{\idx}} |\Psi\rangle. \quad 
\end{align}
The basic idea is that the normalized state $|\bar{\Psi}^\prime \rangle = |\Psi^\prime \rangle / \| \Psi^\prime \|$ is
generated from $|\Psi\rangle$ by a unitary operator $e^{-i \ts \operSB{A}{\idx}}$ acting on a neighborhood of the qubits 
acted on by $\operSB{h}{\idx}$, where $\operSB{A}{\idx}$ can be determined from tomography of 
$|\Psi\rangle$ in this neighborhood up to controllable errors. 
This is illustrated by the simple example where $|\Psi\rangle$ is a product state. The squared norm
\SECONDREVISION{$c = \langle \Psi | e^{-2\ts \operSB{h}{\idx}} | \Psi\rangle$} can be calculated from the expectation value of $\operSB{h}{\idx}$, requiring measurements 
over $k$ qubits,
\begin{align}
\SECONDREVISION{c = 
1 - 2\ts \langle \Psi|  \operSB{h}{\idx} |\Psi\rangle + \mathcal{O}(\ts^2)} \;.
\end{align}
Because $|\Psi \rangle$ is a product state, $|\Psi^\prime \rangle$ is obtained applying the unitary operator 
$e^{-i\ts \operSB{A}{\idx}}$ also on $k$ qubits. $\operSB{A}{\idx}$ can be expanded in terms of an operator basis, 
e.g. the Pauli basis $\{ \oper{\sigma}_i \}$ on $k$ qubits,
\begin{align}
\operSB{A}{\idx} = \sum_{i_1 \ldots i_k} a[\idx]_{i_1  \ldots i_k} \oper{\sigma}_{i_1} \ldots \oper{\sigma}_{i_k}
\SECONDREVISION{\equiv \sum_I a[\idx]_I \oper{\sigma}_I}
. \label{eq:aoperator}
\end{align}
Up to $\mathcal{O}(\ts)$, the coefficients $a[\idx]_{I}$ are defined by
the linear system $\mathbf{S} \mathbf{a}[\idx] = \mathbf{b}$ where the elements of $\mathbf{S}$ and $\mathbf{b}$ 
are expectation values over $k$ qubits, 
\begin{align}
\SECONDREVISION{
S_{I,I^\prime} = \langle \Psi| \oper{\sigma}_{I}^\dag
\oper{\sigma}^{\phantom{\dag}}_{I^\prime}|\Psi\rangle
\,\,,\,\,
b_I  = \frac{-i}{\sqrt{c}} \, \langle \Psi| \oper{\sigma}_{I}^\dag \operSB{h}{\idx} |\Psi\rangle
\,\,.
}
\end{align}
In general, $\mathbf{S}$ has a null space; to ensure $\mathbf{a}[\idx]$ is real, we minimize
$\|  \REVISION{\bar{\Psi}^\prime} - (1-i\ts \operSB{A}{\idx}) \Psi\|^2$ w.r.t. real variations in $\mathbf{a}[\idx]$ (see SI). 
Because the solution is determined from a linear problem, there are no local minima. 

In this simple case, the normalized result of the imaginary time evolution step could be represented by a
unitary update over $k$ qubits, because $|\Psi\rangle$ had correlation length zero. After the initial step, this is no longer the case.
However, for a more general $|\Psi\rangle$ with finite correlations over at most $C$ qubits (i.e.
correlations between observables separated by distance $L$ are bounded by $\exp(- L /C )$),
$| \REVISION{\bar{\Psi}^\prime} \rangle$ can be generated by a unitary acting on a domain of width at most $O(C)$ qubits surrounding
the qubits acted on by $\operSB{h}{m}$. \REVISION{This follows from Uhlmann's theorem~\cite{uhlmann}, which states that
two pure states with marginals close to each other must be related by a unitary transformation on the purifying subsystem (see SI)}. 
The unitary $e^{-i \ts \operSB{A}{m}}$ can then be determined by measurements and solving the least squares problem in this domain (Fig.~\ref{fig:1}). For example, for a nearest-neighbor local Hamiltonian on a $d$-dimension cubic lattice, the domain size $D$ is bounded by $O(C^d)$.
In many physical systems, we expect the maximum correlation length throughout the Trotter steps to increase with $\beta$ and saturate for $C_{\text{max}} \ll N$ \cite{Hastings_CMP_2005}. Fig.~\ref{fig:1}  shows the mutual information between qubits $i$ and $j$ as a function of
imaginary time in the 1D and 2D ferromagnetic transverse field Ising models computed by tensor network simulation (see SI),
demonstrating a monotonic increase and clear saturation. 

The above replacement of imaginary time evolution steps by unitary updates can be extended to more general Hamiltonians,
such as ones with long-range interactions and fermionic Hamiltonians. 
For fermions, in particular, the locality of the corresponding qubit Hamiltonian depends
on the qubit mapping. In principle, a geometric $k$-local fermionic Hamiltonian
can be mapped to a geometric local qubit Hamiltonian~\cite{BRAVYI2002210,Verstraete_JSM_2005}, allowing above techniques to be applied directly.
Alternatively, we conjecture that \SECONDREVISION{by constructing Eq.~\eqref{eq:aoperator} with a local fermionic basis, 
the unitary update can be constructed over a domain size $D \sim O(C^d)$, $C$ being the fermionic correlation length (see SI)}.

\begin{figure*}[t!]
\centering
\includegraphics[width=0.8\textwidth]{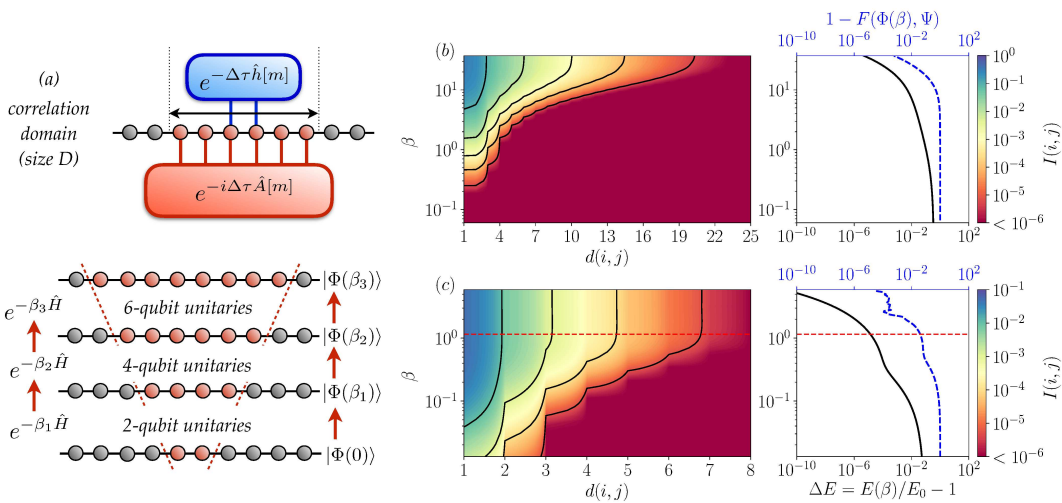} 
\vspace{-0.25cm}
\caption{\SECONDREVISION{
{\bf{Physical foundations of the quantum imaginary time evolution algorithm.}}
}
(a) Schematic of the QITE algorithm.
Top: imaginary-time evolution under a geometric $k$-local operator $\hat{h}[m]$
can be reproduced by a unitary operation acting on $D>k$ qubits.
Bottom: 
exact imaginary-time evolution starting from a product state requires 
unitaries acting on a domain $D$ that grows with correlations.
(b,c) Left: mutual information $I(i,j)$ between qubits $i$, $j$ as a function of 
distance $d(i,j)$ and imaginary time $\beta$, for a 1D (b) and a 2D (c) 
FM transverse-field Ising model, \REVISION{with $h=1.25$, 50 qubits and $h=3.5$, $21 \times 31$ qubits respectively.} 
$I(i,j)$ saturates at longer times.
Right: relative error in the energy $\Delta E$ and fidelity
$F= |\langle \Phi(\beta) | \Psi \rangle |^2$ between
the finite-time state $\Phi(\beta)$ and infinite-time state $\Psi$ as a function
of $\beta$. The noise in the 2D fidelity error at large $\beta$ arises from
the approximate nature of the algorithm used (see SI for details).}
\label{fig:1}
\end{figure*}

\noindent{\bf Cost of QITE}. The number of measurements and classical storage at a given time step (starting propagation from a product state) is bounded by $\exp(O(C^d))$ (with $C$ the correlation length at that time step), since each unitary at that step acts on at most $O(C^d)$ sites; classical solution of the least squares problem has a similar scaling  $\exp(O(C^d))$, as does the synthesis and application as a quantum circuit (composed of two-qubit gates) of the unitary $e^{-i \ts \operSB{A}{m}}$. Thus, space and time requirements are bounded by exponentials in $C^d$, but are
polynomial in $N$ when one is interested in a local approximation of the state (or quasi-polynomial for a global approximation); the polynomial in $N$ comes from the number of terms in $H$ (see SI for details).

The exponential dependence on $C^d$ can be greatly reduced in many cases, \SECONDREVISION{for example if
  $\operSB{A}{m}$ has a locality structure, e.g. if it is (approximately) a $p$-local Hamiltonian
 (i.e. all $a[m]_{i_1 \ldots i_k}$ in Eq.~\eqref{eq:aoperator} are zero except for those where at most $p$ of the $\oper{\sigma}_i$ operators differ from the identity) then the cost of tomography becomes only $C^{O(dp)}$, while the cost of finding and implementing the unitary is $O(p  C^d T_e)$, $T_e$ being the cost of computing one entry of $\operSB{A}{m}$ \cite{berry2015hamiltonian}. If we assume further that $\operSB{A}{m}$ is geometric local, the cost of tomography is reduced further to $O(p C^{d})$.}
\REVISION{However, it is important to note that even if $C$ is too large to construct the unitaries exactly, we can still run the algorithm as a heuristic, truncating the unitary updates to domain sizes that fit the computational budget. This gives the inexact QITE algorithm, described and
studied in detail below.}

Compared to a direct classical implementation of imaginary time evolution, the cost of a QITE time-step (for bounded correlation length $C$) 
is linear in $N$ in space and polynomial in $N$ in time, thus giving an exponential reduction in space and time. 
\SECONDREVISION{Note that a finite correlation length $C_0$ in the ground-state does not generally imply an efficient
  classical strategy. In the SI, we analyze multiple classical heuristics under the assumption of finite ground-state correlations, including:
  truncating the problem size at the ground-state correlation length $C_0$, classical simulation in the Heisenberg representation, and tensor network calculations \cite{VidalTEBD,Schollwock_Annals_2011,Schuch_PRL_2007,PEPShard2018}.
}

\noindent{\bf Inexact QITE}. Given limited resources, for example on near-term devices, we can choose to measure and construct
the unitary over a domain $D$ smaller than induced by correlations, to fit the computational budget. For example, if $D = 1$, this gives a mean-field
approximation of the imaginary time evolution, \REVISION{and larger $D$ gives successively better approximations to the ground-state.}
Importantly, while the unitary is no longer an exact representation of the imaginary time evolution, there is
no issue of a local minimum in its construction, although
the energy is no longer guaranteed to decrease at every step. In this case, one can apply inexact imaginary time evolution
until the energy stops decreasing; the energy will still be a variational upper bound. One can also use the quantum Lanczos algorithm, described later.

\begin{figure*}[t!]
\centering
\includegraphics[width=0.9\textwidth]{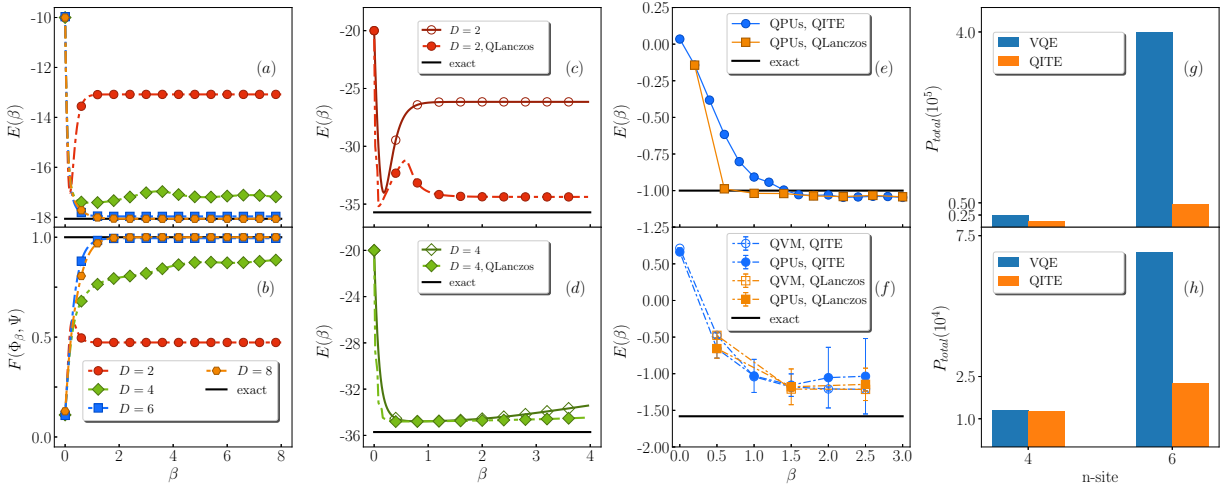}
\vspace{-0.25cm}
\caption{ {\bf{Classical simulation and experimental implementation of QITE and QLanczos algorithms.}}
Column 1: (a) QITE energy $E(\beta)$ and (b) fidelity $F$ between finite-time state
$\Phi(\beta)$ and exact ground state $\Psi$ as function of $\beta$, 
  for a 1D 10-site Heisenberg model, showing convergence with increasing $D$.
Column 2: QITE (dashed lines) and QLanczos (solid lines) energies 
as function of $\beta$, for a 1D Heisenberg model with $N=20$ qubits, using domains 
of $D=2$ (c) and $4$ qubits (d), showing improved convergence of QLanczos over QITE. 
  Column 3: QITE and QLanczos energy as a function of $\beta$ for (e)
  a 1-qubit model and (f) a 2-qubit AFM transverse field Ising model
  using QVMs and QPUs.
  Black lines denote the exact ground-state energy or maximum fidelity.
  \REVISION{Column 4: Estimate of the number of Pauli string expectation values ($P_{total}$) needed for QITE and VQE to converge within (g) 1$\%$ of
    the exact energy for a 4-site (left) and 6-site (right) 1D Heisenberg model with  magnetic field, and (h) 1$\%$ (2$\%$) of the exact energy for a 4-site (6-site) 1D AFM transverse-field Ising model.}
    \SECONDREVISION{Error bars represent standard deviations computed from multiple runs.}}
\label{fig:2}
\end{figure*}

\noindent{\bf Illustrative QITE calculations.} 
To illustrate the QITE algorithm, we have carried out exact classical emulations (assuming perfect
expectation values and gates) for several Hamiltonians (see SI): short-range 1D Heisenberg \REVISION{(with and without a field)};
1D AFM transverse-field Ising; long-range 1D Heisenberg with spin-spin coupling
$J_{ij} =(|i-j|+1)^{-1}$; 1D Hubbard at half-filling;  a 6-qubit MAXCUT 
\cite{Farhi_MIT_2014,Otterbach_arxiv_2017,Moll_QST_2018} instance, and a minimal basis 2-qubit dihydrogen molecular Hamiltonian~\cite{OMalley2015}.
To assess the feasibility of implementation on near-term quantum devices,
we have 
carried out noisy
classical emulation (sampling expectation values and with an error model) using the Rigetti quantum virtual machine (QVM) and a
physical simulation using the Rigetti Aspen-1 QPUs, for a single-qubit field model $(\hat{X}+\hat{Z})/\sqrt{2}$ \cite{lamm2018simulation} and a
1D AFM transverse-field Ising model. \REVISION{We also carried out measurement resource estimates for QITE on the short-range 1D Heisenberg (with field)
model studied in Ref.~\cite{kandala_2017_nature} with VQE, and the 1D AFM transverse-field Ising model; we  compared with resource estimates using the publicly available VQE implementation in IBM's Qiskit.}
We carried out QITE using different fixed domain sizes $D$ for the unitary or fermionic unitary (see SI for descriptions of simulations and models).

Fig.~\ref{fig:2}a-\ref{fig:2}f and \ref{fig:3} show the energy obtained by QITE as a function of $\beta$ and $D$ for the various models. 
As we increase $D$, the asymptotic ($\beta \to \infty$) energies rapidly converge to the exact ground-state. For
small $D$, the inexact QITE tracks the exact QITE for a time until the 
correlation length exceeds $D$. Afterwards, it may go down or up. The non-monotonic behavior is strongest
for small domains; in the MAXCUT example, the smallest domain $D=2$ gives an oscillating energy; the first point
at which the energy stops decreasing is a reasonable estimate of the ground-state energy. 
In all models, increasing $D$ past a maximum value (less than $N$) no longer 
affects the asymptotic energy, showing that the correlations have saturated (this is true even in the MAXCUT instance).
\REVISION{
Figs.~\ref{fig:2}g,~\ref{fig:2}h show an estimate from classical emulation of the number of Pauli string expectation values
to be measured in the QITE algorithm as well as the hardware-efficient VQE ansatz (using the optimization protocol
 in Ref.~\cite{kandala_2017_nature}) to
 obtain an energy accuracy of 1\% in the 1D Heisenberg model with field $J=B=1$ (Fig.~\ref{fig:2}g) and 1\% or 2\% in the 1D AFM transverse-field
 Ising model (Fig.~\ref{fig:2}h, the looser threshold was chosen to enable convergence of VQE). QITE is competitive with VQE for the 4-site model and requires significantly fewer measurements
in the 6-site model. While the number of measurements could potentially be reduced in VQE by different optimizers and Ans\"{a}tze, the data suggests that QITE is a promising alternative to VQE on near-term devices.
}

Figs.~\ref{fig:2}e and \ref{fig:2}f show the results of running the QITE algorithm on Rigetti's QVM and Aspen-1 QPUs for 1- and 2- qubits, respectively.
The error bars are due to gate, readout, incoherent and cross-talk errors. Sufficient samples were used to ensure that sampling error is negligible.
Encouragingly for near-term simulations, despite these errors it is possible to converge to a ground-state energy close to the exact energy for the 1-qubit case. 
This result reflects a robustness that is sometimes informally observed in imaginary time evolution algorithms in which the ground state energy is approached even if the imaginary time step is not perfectly implemented. In the 2-qubit case, although the QITE energy converges,
there is a systematic shift which is reproduced on the QVM using available noise parameters for readout, decoherence and depolarizing noise~\cite{Rigetti}. Remaining discrepancies between the emulator and hardware are likely attributable to cross-talk between parallel gates not included in the noise model (see SI). However, reducing decoherence and depolarizing errors in the QVM or using different sets of
qubits with improved noise characteristics (see SI) all lead to improved convergence to the exact ground-state energy.

\begin{figure}[t!]
\includegraphics[width=\columnwidth]{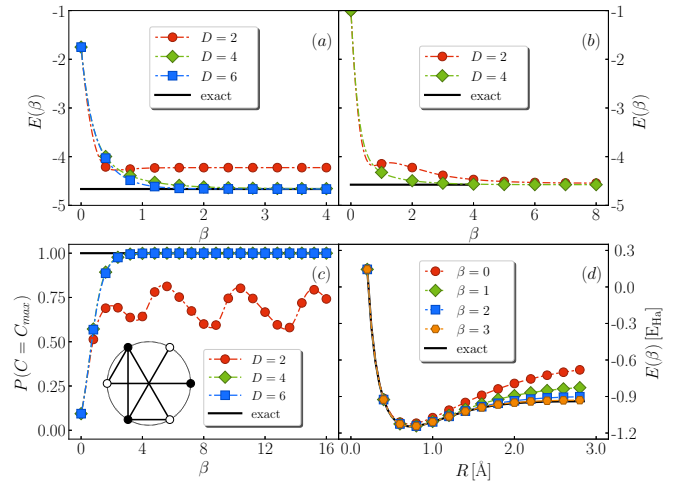} 
\vspace{-0.5cm}
\caption{\SECONDREVISION{{\bf{Application of QITE to long-range spin and fermionic models, and a combinatorial optimization problem.}}} (a) QITE energy as a function of $\beta$ for 
a 6-site 1D long-range Heisenberg model, for unitary domains $D=2-6$;
(b) a 4-site 1D Hubbard model with $U/t = 1$, for unitary domains $D=2,4$. (c) Probability of MAXCUT detection, $P(C=C_{max})$ as a function of imaginary time $\beta$, for the
$6$-site graph in the panel.  (d) QITE energy for
the H$_2$ molecule in the STO-6G basis as a function of bond-length $R$ and $\beta$.
Black line is the exact ground-state energy/probability of detection.
}
\label{fig:3}
\end{figure}

\begin{figure*}[t!]
\includegraphics[width=0.9\textwidth]{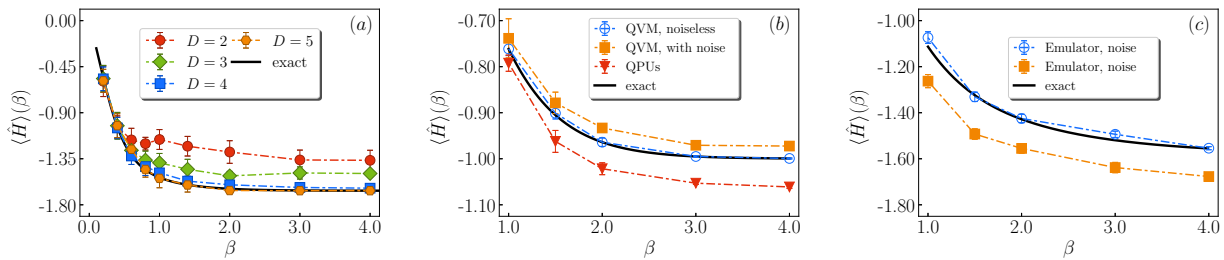}
\vspace{-0.25cm}
\caption{\SECONDREVISION{{\bf{Classical simulation and experimental implementation of the QMETTS algorithm.}}}
Left: Thermal (Gibbs) average $\langle \hat{H} \rangle$ at temperature $\beta$ from QMETTS for a 1D 6-site 
Heisenberg model (exact emulation).
Black line is the exact thermal average without sampling error.
Middle, Right: Thermal average $\langle \hat{H} \rangle$ at temperature $\beta$ from QMETTS for (b) a 1 qubit field model using 
QVMs and QPUs, and (c) 2 qubit AFM transverse field Ising model using QVM. \SECONDREVISION{Error bars represent (block) standard
  deviations computed from multiple samples/runs.}
}
\label{fig:4}
\end{figure*}

\noindent {\bf Quantum Lanczos algorithm}. Given the QITE subroutine, we now consider
how to formulate a quantum Lanczos algorithm, which is an especially
economical realization of a quantum subspace method~\cite{Jarrod1,Jarrod2}. An important practical motivation is that the Lanczos algorithm
typically converges much more quickly than imaginary time evolution, and often in physical simulations only tens of iterations
are needed to converge to good precision. In addition, Lanczos provides a natural way to compute excited states.
Consider the sequence of imaginary time vectors
$|\Phi_l\rangle = e^{-l \ts \oper{H}}|\Phi\rangle$, $l=0, 1, \ldots n$, where $c_l=\|\Phi_l\|$.
In QLanczos, we consider the vectors after even numbers of time steps $|\Phi_0\rangle, |\Phi_2\rangle \ldots$ to form a basis 
for the ground-state. (SI describes the equivalent treatment
in terms of normalized imaginary time vectors). These vectors define an overlap matrix whose elements can be computed entirely from norms, $S_{ll'} 
= \langle \Phi_l |\Phi_{l'}\rangle = c^2_{(l+l')/2}$, where $c_{(l+l')/2}$
is the norm of another integer time step vector, and the overlap matrix elements for $n/2$ vectors can be accumulated 
for free after $n$ steps of time evolution. 
The Hamiltonian matrix elements satisfy the identity $H_{ll'} = \langle \Phi_l | \oper{H} | \Phi_{l'}\rangle 
= \langle \Phi_{(l+l')/2} | \oper{H} | \Phi_{(l+l')/2}\rangle$.
 Although the Hamiltonian has $\sim n^2$ matrix elements \REVISION{in the basis of the $\Phi_l$ states}, 
 there are only $\sim n$ unique elements, and importantly, 
 each is a simple expectation value of the energy during the imaginary time evolution.  
This economy of matrix elements is a property shared with the classical Lanczos algorithm. Whereas the 
classical Lanczos iteration builds a Krylov space in powers of $\oper{H}$, QLanczos builds a Krylov space in powers of 
$e^{-2 \ts \oper{H}}$; in the limit of small $\ts$ these Krylov spaces are identical. Diagonalization of the QLanczos Hamiltonian
matrix is guaranteed to give a ground-state energy lower than that of the last imaginary time vector $\Phi_n$ (while
higher roots approximate excited states).

With a limited computational budget, we can use inexact QITE to generate $\Phi_l$, $\Phi_l'$. However, in this
case the above expressions for $S_{ll'}$ and $H_{ll'}$ in terms of expectation values no longer exactly hold,
which can create numerical issues (e.g. the overlap may no longer be positive). 
To handle this, as well as errors due to noise and sampling in real experiments, the QLanczos algorithm
needs to be stabilized by ensuring that successive vectors are not nearly linearly dependent (see SI).

We demonstrate the QLanczos algorithm using classical emulation on the 1D Heisenberg Hamiltonian,
as used for the QITE algorithm in Fig. \ref{fig:2} (see SI). 
Using exact QITE (large domains) to generate  matrix elements, exact quantum Lanczos converges much 
more rapidly than imaginary time evolution. 
\REVISION{Convergence of inexact QITE (small domains), however, can both be faster and reach lower energies than inexact quantum Lanczos.}
We also assess the feasibility of QLanczos in  presence of noise, using emulated noise on the Rigetti QVM as well as
on the Rigetti Aspen-1 QPUs.
In Fig. \ref{fig:2}, we see
that QLanczos also provides more rapid convergence than QITE with both noisy classical emulation as well as on the physical device
for 1 and 2 qubits.

\noindent {\bf Quantum thermal averages}. The QITE subroutine can be used in a range of other algorithms. For
example, we discuss how to compute thermal averages $\mathrm{Tr}\big[ \oper{O}  e^{-\beta \oper{H}} \big]
/ \mathrm{Tr} \big[ e^{-\beta \oper{H}} \big]$ using imaginary time evolution.
Several procedures have been proposed for quantum thermal averaging, ranging from generating the 
finite-temperature state explicitly by equilibration with a bath~\cite{Terhal_PRA_2000}, to a quantum analog of
Metropolis sampling~\cite{Temme_Nature_2011} that relies on phase estimation, as well as methods based
on ancilla based Hamiltonian simulation with post-selection~\cite{Somma} and approaches based on recovery maps \cite{brandao2019finite}. However, given a method for imaginary time evolution, one can generate thermal averages 
of observables without any ancillae or deep circuits. This can be done by adapting to the quantum setting the classical minimally entangled typical thermal 
state (METTS) algorithm \cite{White_PRL_2009,Miles_NJP_2010}, which generates a Markov chain from which the thermal average can be sampled.
The QMETTS algorithm can be carried out as follows (i) start from a product state, carry out imaginary-time 
evolution (using QITE) up to time $\beta$ (ii) measure the expectation value of $\hat{O}$ to produce its thermal average (iii) measure a product operator such as $\oper{Z}_1 \oper{Z}_2 \ldots \oper{Z}_N$, 
to collapse back onto a random product state (iv) repeat (i). Note that in step (iii) one can measure in any product basis, and randomizing 
the product basis can be used to reduce the autocorrelation time and avoid ergodicity problems in sampling.
In Fig. \ref{fig:4} we show the results of  quantum METTS (using exact classical emulation) for the thermal average $\langle \oper{H}\rangle$
as a function of temperature $\beta$,
for the 6-site Heisenberg model for several temperatures 
and domain sizes; sufficiently
large $D$ converges to the exact thermal average at each $\beta$; error bars reflect only finite QMETTS samples.
We also show an implementation of quantum METTS
on the Aspen-1 QPU and QVM with a 1-qubit field model (Fig.~\ref{fig:4}b),
and using the QVM for a 2-qubit AFM transverse field Ising model (Fig.~\ref{fig:4}c).

\SECONDREVISION{
\noindent {\bf Conclusions} 
In summary, the quantum analogs of imaginary-time evolution,
Lanczos and METTS algorithms we have presented enable a new class of eigenstate and thermal state quantum simulations, that can be carried out without ancillae or deep circuits and that, for bounded correlation length, achieve exponential reductions in space and time per iteration relative to known classical counterparts.
Encouragingly, these algorithms appear useful in conjunction with near-term quantum architectures, and serve to demonstrate the power of quantum elevations of classical simulation techniques, in the continuing search for quantum supremacy.
}

{\bf{Acknowledgments.}} MM, GKC, FGSLB, ATKT, AJM were supported by the US NSF via RAISE-TAQS CCF 1839204. 
         MJO'R was supported by an NSF graduate fellowship via grant No. DEG-1745301; 
         the tensor network algorithms were developed with the
         support of the US DOD via MURI FA9550-18-1-0095. EY was supported by a Google fellowship.
         CS was supported by the US DOE via DE-SC0019374.
         GKC is a Simons Investigator in Physics and a member of the Simons Collaboration on the Many-Electron Problem.
         The Rigetti computations were made possible by a generous grant through
         Rigetti Quantum Cloud services supported by the CQIA-Rigetti Partnership Program. We thank GH Low, JR McClean, R Babbush for discussions, and the Rigetti team for help with the QVM and QPU simulations.
         
{\bf{Author contributions and data availability.}}

MM, CS, GKC designed the algorithms. FGSLB established the mathematical proofs and error estimates.
EY and MJO'R performed classical tensor network simulations. MM, CS, ATKT carried out classical exact emulations. ATKT and AJM designed and
carried out the Rigetti QVM and QPU experiments. All authors contributed to the discussion of results and writing of the manuscript.
The code used to generate the data presented in this study can be publicly accessed on GitHub at $\mathsf{https://github.com/mariomotta/QITE.git}$

\appendix

\section{SUPPLEMENTAL INFORMATION}

\subsection{Representing imaginary-time evolution by unitary maps}

\SECONDREVISION{
In this section, we discuss how to emulate imaginary time evolution by measurement-assisted unitary circuits acting on suitable domains.
}
As discussed in the main text, we  map the scaled non-unitary action of $e^{-\Delta\tau \operSB{h}{l}}$
  on a state $\Psi$ 
to that of a unitary $e^{-i\Delta\tau \operSB{A}{l} }$, i.e.
\begin{align}
| \REVISION{\bar{\Psi}}^\prime \rangle \equiv c^{-1/2} \, e^{-\Delta \tau \operSB{h}{l} } |\Psi\rangle = e^{-i\Delta\tau \operSB{A}{l}} |\Psi\rangle \quad .
\end{align}
where $c = \langle \Psi | e^{-2 \Delta \tau \operSB{h}{l} } |\Psi\rangle$.
$\operSB{h}{l}$ acts on $k$ qubits; 
$\hat{A}$ is Hermitian and acts on a domain of $D$ qubits around the support of $\operSB{h}{l}$,
and is expanded as a sum of Pauli strings acting on the $D$ qubits,
\begin{equation}
\operSB{A}{l} 
= \sum_{i_1i_2 \ldots i_D} a[l]_{i_1i_2 \ldots i_D} \hat{\sigma}_{i_1}\hat{\sigma}_{i_2} \ldots \hat{\sigma}_{i_D} 
= \sum_I a[l]_I \hat{\sigma}_I 
\;,
\label{eq:pauli}
\end{equation}
where $I$ denotes the index $i_1i_2 \ldots i_D$. Define
\begin{equation}
|\Delta_0\rangle = \frac{| \REVISION{\bar{\Psi}}^\prime \rangle - | \Psi\rangle}{\Delta \tau}
\;,\;
|\Delta\rangle = -i \operSB{A}{l} |\Psi\rangle
\;.
\end{equation}
Our goal is to minimize the difference $||\Delta_0 - \Delta||$. If the unitary $e^{-i\Delta\tau \operSB{A}{l} }$ is defined over a 
sufficiently large domain $D$ (related to the correlation length of $|\Psi\rangle$, 
see Section~\ref{sec:spread}) then this error minimizes at $\sim 0$, for small $\Delta \tau$. Minimizing for real $a[l]$
corresponds to minimizing the quadratic function $f(a[l])$
\begin{align}
f(a[l]) = f_0 + \sum_I b_I a[l]_I + \sum_{IJ} a[l]_I S_{IJ} a[l]_J
\end{align}
where
\begin{align}
f_0 &= \langle \Delta_0 | \Delta_0 \rangle \quad , \\
S_{IJ} &= \langle \Psi | \hat{\sigma}^\dag_I \hat{\sigma}_J | \Psi\rangle \quad ,\\
b_I &= i \, \langle \Psi | \hat{\sigma}^\dag_I | \Delta_0 \rangle -  i \, \langle  \Delta_0 | \hat{\sigma}_I | \Psi \rangle \quad ,
\end{align}
whose minimum obtains at the solution of the linear equation
\begin{align}
\left( \mathbf{S}+\mathbf{S}^T \right) \mathbf{a}[l] = -\mathbf{b} \label{eq:lineareq}
\end{align}
In general, $\mathbf{S}+\mathbf{S}^T$ may have a non-zero null-space. Thus, we solve Eq.~\eqref{eq:lineareq}
either by applying the generalized inverse of $\mathbf{S}+\mathbf{S}^T$ or by an iterative algorithm such as conjugate gradient.

\SECONDREVISION{Note that the above results for geometric local
  Hamiltonians can be extended to Hamiltonians with long-range terms.
  The primary difference is a modification of the domain over
  which the unitary acts.
  For example, for a spin Hamiltonian with long-range pairwise terms, the scaled action of $e^{-\Delta \tau \hat{h}[l]}$ (if $\hat{h}[l]$ acts on qubits $i$ and $j$) can be emulated, to within accuracy $\varepsilon$, by a unitary constructed in the neighborhoods of $i$ and $j$, with the domain
  size given by the result in Eq. (\ref{eq:dbound}), see also discussion below).}
\REVISION{
  For fermionic Hamiltonians, we replace the Pauli operators in Eq.~(\ref{eq:pauli}) by fermionic field operators, i.e. $\sigma \in \{1, f, f^\dag, f^\dag f\}$, and conjecture that the analogous result holds for the domains
  for the fermionic operators as for spin operators.
  For a number conserving Hamiltonian, such as the fermionic Hubbard Hamiltonian treated in Fig. 3 in the main text, 
  we retain only those terms with an equal number of creation and annihilation operators, to conserve particle number.
}
 
\REVISION{
\subsection{Real Hamiltonians and states}
As the cost
to construct the quantities ${\bf{S}}$, 
${\bf{b}}$ and to solve the linear system \eqref{eq:lineareq} increases
exponentially with $D$, it is natural to seek out simplifications that can be made
rigorously. For example, as described above, to emulate the non-unitary action of a number-conserving fermionic Hamiltonian,
we can consider $\hat{A}[l]$ to contain only fermionic operator strings that also conserve particle number. In the main text,
we also considered the case where $\hat{A}[l]$ is itself approximately $p$-local, which removes the exponential dependence on $D$.

Another common scenario concerns Hamiltonians and states that only
have real matrix elements and coefficients in the $Z$ computational basis.
Then, since $\langle \Psi|$ and $|\Delta_0\rangle$ are real in the computational basis and
\begin{equation}
\begin{split}
b_I 
&= -2 \, \mbox{Im} \left[ \langle \Psi | \hat{\sigma}_I^\dagger | \Delta_0 \rangle \right],
\end{split}
\end{equation}
$b_I \equiv 0$ unless the matrix elements of $\hat{\sigma}_I^\dagger$ have 
 non-zero imaginary part. When the Pauli basis is used, this means that $b_I \equiv 0$
unless $\hat{\sigma}_I^\dagger$ contains an odd number of $\hat{Y}$ operators.

The number of such operators for a domain size $D$, $y(D)$, is $2^D \frac{2^D-1}{2} \simeq \frac{4^D}{2}$.
This can be shown by induction over $D$. For $D=1$, one has 
\begin{equation}
y(D=1)=1
\end{equation}
Pauli strings with an odd number of $\hat{Y}$'s (i.e. just $\hat{Y}$).
For $D=2$, one has 
\begin{equation}
y(D=2)=6
\end{equation}
such strings, namely
\begin{equation}
\hat{I} \hat{Y}, \hat{X} \hat{Y}, \hat{Z} \hat{Y}, \hat{Y} \hat{I},
\hat{Y} \hat{X}, \hat{Y} \hat{Z} \quad .
\end{equation}
For $D=3$, the number grows to 
\begin{equation}
y(D=3)=28
\end{equation}
with strings
\begin{equation}
\begin{array}{lll}
\hat{I} \hat{I} \hat{Y}, & \dots, & \hat{I} \hat{Y} \hat{Z} \\
\hat{X} \hat{I} \hat{Y}, & \dots, & \hat{X} \hat{Y} \hat{Z} \\
\hat{Y} \hat{I} \hat{I}, & \dots, & \hat{Y} \hat{Z} \hat{Z} \\
\hat{Z} \hat{I} \hat{Y}, & \dots, & \hat{Z} \hat{Y} \hat{Z} \\
\end{array}
\end{equation}
and so on.
Pauli strings of length $D+1$ containing an odd number of $\hat{Y}$ operators are obtained either by attaching a $\hat{Y}$ operator to a length-$D$ string containing an even number of $\hat{Y}$ operators,
or by attaching a $\hat{I}$, $\hat{X}$, $\hat{Z}$ operator to a length-$D$ string containing 
an odd number of $\hat{Y}$ operators. Therefore,  $O(D)$  obeys the recursion
relation 
\begin{equation}
y(D+1) = 3 y(D) + ( 4^D - y(D) ) \quad.
\end{equation}
This recursion relation is solved, with the initial condition
$y(D=1)=1$, by 
\begin{equation}
y(D) = 2^D \frac{2^D-1}{2} \quad.
\end{equation}

This result means that both $\mathbf{b}$ and $\mathbf{S}+{\mathbf{S}}^T$ can be assembled from $y(D)$ Pauli string expectation values, roughly
half the number of measurements needed if one did not assume real Hamiltonians and states. Further, the dimension
of $\mathbf{b}$ and $\mathbf{S}+{\mathbf{S}}^T$ is $y(D)$ and $y(D) \times y(D)$ respectively. Asymptotically,
this reduces the cost of solving the linear system  Eq. \eqref{eq:lineareq} by a factor of $1/8$, assuming dense matrix techniques.
}
\subsection{Rigorous run time bounds}
\label{sec:spread}

Here we present a more detailed analysis of the running time of the algorithm. Consider a $k$-local Hamiltonian 
\begin{equation}
H = \sum_{l=1}^m \operSB{h}{l}
\;,
\end{equation}
acting on a $d$-dimensional lattice with $\Vert h_i \Vert \leq 1$, where $\Vert * \Vert$ is the operator norm. 
\REVISION{
  Note that, if a quantum chemistry system is studied with an orthonormal basis of spatially localized states,
  such states can be approximately positioned on a lattice, and the results
 of this section apply on length scales larger than the size of the employed basis
functions.}
In imaginary-time evolution one typically applies Trotter formulae to approximate 
\begin{equation}
\frac{e^{- \beta \oper{H} } | \Psi_0 \rangle}{ \Vert   e^{- \beta \oper{H} } | \Psi_0 \rangle \Vert   }
\simeq
\frac{  \left (  e^{- \dt \operSB{h}{1} } \ldots  e^{- \dt \operSB{h}{m} }  \right)^{n} | \Psi_0 \rangle   }   { \Vert  \left (  e^{- \dt \operSB{h}{1} } \ldots  e^{- \dt \operSB{h}{m} } \right)^{n} | \Psi_0 \rangle   \Vert  }
\;.
\label{trotterdecomp}
\end{equation}
for an initial state $| \Psi_0 \rangle$ (which we assume to be a product state).
This approximation leads to an error which can be made as small as one wishes by increasing the number of time steps $n$.
Let $| \Psi_s \rangle$ be the state (after renormalization) obtained by applying $s$ terms $e^{- \dt \operSB{h}{i} }$ from $\big( e^{- \dt \operSB{h}{1} } \ldots  e^{- \dt \operSB{h}{m} } \big)^{n}$; 
with this notation $| \Psi_{mn} \rangle$ is the state given by Eq. (\ref{trotterdecomp}). In the QITE algorithm, instead of applying each of the operators $e^{- \dt \operSB{h}{i} }$ to $| \Psi_0 \rangle$ 
(and renormalizing the state), one applies local unitaries $\oper{U}_s$ which should approximate the action of the original operator. Let $| \Phi_s \rangle$ be the state after $s$ unitaries have been applied. 

Let $C$ be an upper bound on the correlation length of $| \Psi_s \rangle$ for every $s$: we assume that for every $s$, and every pair of observables 
$\oper{A}$ and $\oper{B}$ acting on domains separated by $\text{dist}(A,B)$ sites, 
\begin{equation} \label{correlationdecay}
\begin{split}
C_s(\oper{A},\oper{B}) &= \langle \Psi_s | \oper{A} \otimes \oper{B} | \Psi_s \rangle - \langle \Psi_s | \oper{A}  | \Psi_s \rangle \langle \Psi_s | \oper{B} | \Psi_s \rangle \\
&\leq \Vert \oper{A} \Vert \Vert \oper{B} \Vert e^{- \text{dist}(A, B) / C}.
\end{split}
\end{equation}

\begin{theorem}  \label{bounderrors}
For every $\varepsilon > 0$, there are unitaries $\hat{U}_s$ each acting on 
\begin{equation}
N_q = k \, (2 C)^d \, \ln^d\left( 2 \sqrt{2} \, n m \, \varepsilon^{-1} \right)
\end{equation}
qubits, such that
\begin{equation}
\left \Vert   | \Psi_{mn} \rangle  -  | \Phi_{mn} \rangle  \right \Vert \leq \varepsilon  
\;.
\end{equation}
\end{theorem}

\begin{proof}
We have
\begin{eqnarray} \label{boundingerror1}
&\left \Vert  | \Psi_{s} \rangle  -| \Phi_{s} \rangle   \right  \Vert =
\left \Vert  | \Psi_{s} \rangle   - \hat{U}_s | \Phi_{s-1} \rangle   \right  \Vert \nonumber \\
\leq &\left \Vert   | \Psi_{s} \rangle   - \hat{U}_s | \Psi_{s-1} \rangle   \right  \Vert    +  \left \Vert | \Psi_{s-1} \rangle  - | \Phi_{s-1} \rangle \right \Vert.    
\end{eqnarray}
To bound the first term we use our assumption that the correlation length of $| \Psi_{s-1} \rangle$ is smaller than $C$. Consider a region $R_{v}$ of all sites that are a distance at most $v$ (in the Manhattan distance on the lattice) of the sites in which $h_{i_s}$ acts. Let $\text{tr}_{\backslash R_v}(| \Psi_s \rangle \langle \Psi_s | )$ be the reduced state on $R_v$, obtained by partial tracing over the complement of $R_v$ in the lattice. Since
\begin{equation}
 | \Psi_{s} \rangle = \frac{ e^{- \dt \operSB{h}{i_s} }  | \Psi_{s-1} \rangle }{ \Vert e^{ - \dt \operSB{h}{i_s} } | \Psi_{s-1} \rangle \Vert },
\end{equation}
it follows from Eq. (\ref{correlationdecay}) and Lemma 9 of \cite{brandao2015exponential} that 
\begin{equation} \label{boundmarginal}
\begin{split}
&\left \Vert \text{tr}_{\backslash R_v}(| \Psi_s \rangle \langle \Psi_s | ) -  \text{tr}_{\backslash R_v}(| \Psi_{s-1} \rangle \langle \Psi_{s-1} | ) \right \Vert_1 \\
&\leq \Vert e^{ \dt \operSB{h}{i_s} } \Vert^{-1} e^{- \frac{v}{C}} \leq 2 e^{- \frac{v}{C}},
\end{split}
\end{equation}
where we used that for $n \geq 2\beta$, 
\begin{equation}
\Vert e^{- \dt \operSB{h}{i_s} } \Vert \geq \Vert I - \dt \operSB{h}{i_s} \Vert \geq 1 - \dt \geq 1/2
\;.
\end{equation}
Above, $\Vert * \Vert_1$ is the trace norm.
The key result in our analysis is Uhlmann's theorem (see e.g. Lemmas 11 and 12 of \cite{brandao2015exponential}). 
It states that two pure states with nearby marginals must be related by a unitary on the purifying system. 
\REVISION{
In more detail, if $| \eta \rangle_{AB}$ and $| \nu \rangle_{AB}$ are two states with partial traces
$| \eta \rangle_A$ and $| \nu \rangle_A$ over the complement of $A$, such that
$\Vert | \eta \rangle_A - | \nu \rangle_A \Vert_1 \leq \delta$,
then there exists a unitary $\oper{V}$ acting on $B$ such that} 
\begin{equation} 
\label{uhlmannstatement}
\Vert | \eta \rangle_{AB} - (I \otimes \oper{V} ) | \nu \rangle_{AB}    \Vert \leq 2 \sqrt{\delta}.
\end{equation}
Applying Uhlmann's theorem to $| \Psi_s \rangle$ and $| \Psi_{s-1} \rangle$, with $B = R_v$, and using Eq. (\ref{boundmarginal}), we find that there exists a unitary $\oper{U}_s$ acting on $R_{v}$ s.t. 
\begin{equation} 
\label{uhlmannbound}
 \left \Vert  | \Psi_{s} \rangle   - \oper{U}_s | \Psi_{s-1} \rangle   \right  \Vert \leq 2 \sqrt{2} \, e^{- \frac{v}{2C}},
\end{equation}
which by Eq. (\ref{boundingerror1}) implies 
\begin{equation}
 \left \Vert  | \Psi_{nm} \rangle  - | \Phi_{nm} \rangle   \right  \Vert \leq 2 \sqrt{2} \, m n \, e^{- \frac{v}{2C}},
\end{equation}

Choosing $\nu = 2 C \ln(2 \sqrt{2} n m \varepsilon^{-1})$ as the width of the support of the approximating unitaries, the error term above is $\varepsilon$. 
The support of the local unitaries is $k \nu^d$ qubits (as this is an upper bound on the number of qubits in $R_\nu$). Therefore each unitary $\oper{U}_s$ 
acts on at most
\begin{equation}
N_q = k \, (2 C)^d \ln^d\left( 2 \sqrt{2} \, n m \, \varepsilon^{-1} \right) \label{eq:dbound}
\end{equation}
qubits. 
\end{proof}

\vspace{0.4 cm}

\noindent \textit{Finding $\oper{U}_s$:} In the algorithm we claim that we can find the unitaries $\oper{U}_s$ by solving a least-square problem. 
This is indeed the case if we can write them as $\oper{U}_s = e^{i \frac{\operSB{A}{s}}{n} }$ with $\operSB{A}{s}$ a Hamiltonian of constant norm. 
Then for sufficiently large $n$, $\oper{U}_s = I + i \frac{\operSB{A}{s}}{n}  + O\left( \frac{1}{n^2} \right)$ and we can find $\operSB{A}{s}$ by 
performing tomography of the reduced state over the region where $\oper{U}_s$ acts and solving the linear problem given in the main text. 
Because we apply Uhlmann's Theorem to 
\begin{equation}
| \Psi_{s-1} \rangle
\;\;
\mbox{and}
\;\;
\frac{ e^{- \dt \operSB{h}{i_s} }  | \Psi_{s-1} \rangle }{ \Vert e^{ - \dt \operSB{h}{i_s} } | \Psi_{s-1} \rangle \Vert }
\;\;,
\end{equation}
using $e^{ - \dt \operSB{h}{i_s} } = I - \dt \operSB{h}{i_s} + O\left( \frac{1}{n^2} \right)$ and following the proof of the Uhlmann's Theorem, 
we find that the unitary can indeed be taken to be close to the identity, i.e. $\oper{U}_s$ can be written as $e^{i \frac{\operSB{A}{s}}{n} }$.

\vspace{0.4 cm}

\noindent \textit{Total Run Time:} Theorem \ref{bounderrors} gives an upper bound on the maximum support of the unitaries needed for a Trotter update, 
while tomography of local reduced density matrices gives a way to find the unitaries. 
The cost for tomography is quadratic in the dimension of the region, so it scales as $\exp(O( N_q ))$. 
This is also the cost to solve classically the linear system which gives the associated Hamiltonian $\operSB{A}{s}$ and of finding a circuit decomposition 
of $\oper{U}_s = e^{i \frac{\operSB{A}{s}}{n}}$ in terms of two-qubit gates. 
As this is repeated $mn$ times, for each of the $mn$ terms of the Trotter decomposition, the total running time (of both quantum and classical parts) is
\begin{equation}
T = m n \, e^{O(N_q)} = m n \, e^{O\left( k \, (2 C)^d \ln^d\left( 2 \sqrt{2} \, n m \, \varepsilon^{-1} \right) \right)}
\;.
\end{equation}
This is exponential in $C^d$, with $C$ the correlation length, and quasi-polynomial in $n$ (the number of Trotter steps) and $m$ (the number of local 
terms in the Hamiltonian. Note that typically $m = O(N)$, with $N$ the number of sites). While this an exponential improvement over the $\exp(O(N))$ 
scaling classically, the quasi-polynomial dependence on $m$ can still be prohibitive in practice. Below we show how to improve on that.

\vspace{0.4 cm}

\noindent \textit{Local Approximation:} If one is only interested in a local approximation of the state (meaning that all the local marginals of $|\Phi_{nm} \rangle$ 
are close to the ones of $e^{- \beta \oper{H} } |\Psi_0 \rangle$, but not necessarily the global states), then the support of the unitaries becomes 
independent of the number of terms of the Hamiltonian $m$ (while for global approximation we have a polylogarithmic dependence on $m$):

\begin{theorem} 
\label{bounderrorslocal}
For every $\varepsilon > 0$, there are unitaries $\oper{U}_s$ each acting on

\begin{equation}
N_q = k (2 C)^d \ln^d \left( 2 \sqrt{2} n \left(|S| + C \ln^d\left( 8 n C(2C)^{d+1} \varepsilon^{-1} \right) \right)\right)
\end{equation}
qubits such that, for every connected region $S$ of size at most $|S|$,
\begin{equation} 
\label{localapprox}
\left \Vert \mathrm{tr}_{\backslash S} ( | \Psi_{mn} \rangle \langle \Psi_{mn} |  ) - \mathrm{tr}_{\backslash S} ( | \Phi_{mn} \rangle \langle \Phi_{mn} |  )  \right \Vert_1 \leq \varepsilon
\;.
\end{equation}
\end{theorem}

\begin{proof}
Consider the unitaries $\oper{U}_s$ obtained in the proof of Theorem \ref{bounderrors} satisfying Eq. (\ref{uhlmannbound}).

Consider the replacement of the local term of the Trotter expansion by the unitary $\hat{U}_s$ for all local terms which are more than $2C \log(1/\delta)$ sites away from the region $S$. 
Because the correlation length is always smaller than $C$, we find by Lemma 9 of \cite{brandao2015exponential} that the total error $\eta$ on the reduced density matrix in region $S$ 
can be bounded as
\begin{equation}
\eta = n \int_{2C \ln(1/ \delta)}^{\infty} e^{- l / 2C} l^d dl \leq 4 n C (2C)^{d+1} \delta \;.    
\end{equation}
For the local terms which are at most a distance $2C \log(1/\delta)$ from the region $S$, in turn, the total error is bounded by the sum of each individual term, giving
\begin{equation}
\eta = (|S| + C \log(1/\delta))^d n 2 \sqrt{2} e^{-\frac{\nu}{2C}}
\;.
\end{equation}

Choosing 
\begin{equation}
\begin{split}
\delta &= \frac{\varepsilon}{(8 n C (2C)^{d+1})} \\
\nu &= 2 C \ln\left( 2 \sqrt{2} n \left( |S| + C \ln\left(8 n C(2C)^{d+1} \varepsilon^{-1} \right)^d \right)\right) \\
\end{split}
\end{equation}
gives the result. 
\end{proof}

\vspace{0.4 cm}

 \noindent \textit{Non-local Terms:} Suppose the Hamiltonian has a term $\operSB{h}{q}$ acting on qubits which are not nearby, e.g. on two sites $i$ and $j$. Then $e^{- \dt \operSB{h}{q}}$ 
 can still be replaced by a unitary, which only acts on sites $i$  and $j$ and qubits in the neighborhoods of the two sites. This is the case if we assume that the state has a finite correlation 
 length and the proof is again an application of Uhlmann's theorem (we follow the same argument from the proof of Theorem \ref{bounderrors}  but define $R_v$ in that case as the union of 
 the neighborhoods of $i$ and $j$). Note however that the assumption of a finite correlation length might be less natural for models with long range interactions.  
 
 \vspace{0.2 cm}
 
 \noindent \textit{Scaling with temperature and increase of correlation length:} Our discussion has been based on the assumption that the correlation length $C$ is small on all intermediate states. 
 Here we discuss the range of validity of the assumption. 
 
 Let us begin with an example where the correlation length can increase very quickly with number of local terms applied (this was communicated to us by Guang Hao Low). 
 Consider a projection on two qubits $\oper{P}_{i, i+1} = |0, 0 \rangle \langle 0, 0 |_{i,i+1} + |1, 1 \rangle \langle 1, 1|_{i, i+1}$. Then  
 \begin{equation}
\oper{P}_{1, 2} \oper{P}_{2, 3} \ldots \oper{P}_{n-1, n} | + \rangle^{\otimes n}, 
\end{equation}
 with $|+ \rangle = (|0 \rangle + |1 \rangle)/ \sqrt(2)$, is the GHZ state $(|0 \ldots 0 \rangle + |1 \ldots 1 \rangle)/\sqrt{2}$, which has correlation length $C = n$. 
 While the projector $\oper{P}_{i, i+1}$ cannot appear as a local term $e^{- \dt \operSB{h}{i} }$  in the Trotter decomposition, this example show that we cannot expect 
 a speed-of-sound bound on the spread of correlations for a circuit with non-unitary gates; indeed the example shows a depth two circuit can already create long range correlations. 
 
 However, we expect that generically the correlations do grow ballistically. Consider the state 
 \begin{equation} 
 |\psi_{n} \rangle := \frac{  \left ( e^{- \dt \operSB{h}{1} } \ldots  e^{- \dt \operSB{h}{m} }  \right)^{n} | \Psi_0 \rangle   }   { \Vert  \left ( e^{- \dt \operSB{h}{1} } \ldots  e^{- \dt \operSB{h}{m} } \right)^{n} | \Psi_0 \rangle   \Vert  }.
\end{equation}
after $n$ rounds have been applied. Let us assume the Hamiltonian acts on a line, is translation invariant and has nearest-neighbor interactions. 
Then the state is a matrix product state of bond dimension at most $2^{n}$. For matrix product states we can bound the correlations as follows 
(see e.g. Lemma 22 of \cite{brandao2015exponential})
\begin{equation} 
\begin{split}
C_s(\oper{A},\oper{B}) &= \langle \Psi_s | \oper{A} \otimes \oper{B} | \Psi_s \rangle - \langle \Psi_s | \oper{A} | \Psi_s \rangle \langle \Psi_s | \oper{B} | \Psi_s \rangle \\
&\leq \Vert \oper{A} \Vert \Vert \oper{B} \Vert 2^{2n} e^{ - \Delta \text{dist}(A, B)} \;.
\end{split}
\end{equation}
where we define the gap of the matrix-product-state as $\Delta := 1 - \lambda$, with $\lambda$ the second largest eigenvalue of the transfer matrix of the matrix product state 
(normalized so that the largest eigenvalue is one). In the GHZ example above, the gap $\Delta = 0$ and that is the reason for the fast build up of correlations. 
Typically we expect the gap to be independent of $n$ or decrease mildly as $1/\text{poly}(n)$. 

From the above, we can replace a non-unitary local Trotter term applied to $|\psi_{n} \rangle$ by an unitary acting on $O(n / \Delta)$ qubits. 
Taking $n = O(\beta)$ to reach temperature $\beta$ in the imaginary time evolution, the support of the unitaries would scale as $O(\beta/\Delta)$.
Assuming $\Delta$ is a constant, we find a linear increase in temperature. 

We also expect the linear growth of correlations/unitary support with inverse temperature also to hold generically in two dimensions, 
although there the analysis is more subtle as rigorous results for the expected behavior of the transfer operator 
(which becomes a one-dimensional tensor product operator) and its gap are not available.
 
\subsection{Spreading of correlations}

In the main text, we argued that the correlation volume $V$ of the state $e^{-\beta H}|\Psi\rangle$ is bounded
for many physical Hamiltonians and saturates at the ground-state with $V \ll N$ where $N$ is the system size.
To numerically measure correlations, we use the mutual information between two sites, defined as
\begin{align}
I(i,j) =   S(i)+S(j) - S(i,j)
\end{align}
where $S(i)$ is the von Neumann entropy of the density matrix of site $i$ ($\rho(i)$) and similarly for $S(j)$, and $S(i,j)$
is the von Neumann entropy of the two-site density matrix for sites $i$ and $j$ ($\rho(i,j)$).

To compute the mutual information in Fig. 1 in the main text, we used matrix product state (MPS) and finite projected entangled pair state 
(PEPS) imaginary time evolution for the spin-$1/2$ 1D and 2D FM transverse field Ising model (TFI)
\begin{align}
\oper{H}_{TFI} = - \sum_{\langle ij \rangle} \oper{Z}_i \oper{Z}_j  - h \sum_{i} \oper{X}_i
\end{align}
where the sum over $\langle i, j \rangle$ pairs are over nearest neighbors. We use the parameter $h=1.25$ for the 1-D calculation and $h=3.5$ 
for the 2-D calculations as the ground-state is gapped in both cases.  It is known that the ground-state correlation length is finite.

\noindent \textbf{MPS}.
We performed MPS imaginary time evolution (ITE) on a 1-D spin chin with $L=50$ sites with open boundary conditions. 
We start from an initial state that is a random product state, and perform ITE using time evolution block decimation 
(TEBD) \cite{vidal2004TEBD,schollwock2011mps} with a first order Trotter decomposition. In this algorithm, 
the Hamiltonian is separated into terms operating on even and odd bonds. The operators acting on a single bond are exponentiated exactly.
One time step is given by time evolution of odd and even bonds sequentially, giving rise to a Trotter error on the order of the time step $\Delta \tau$. 
In our calculation, a time step of $\Delta \tau = 0.001$ was used. 

We carry out ITE simulations with maximum bond dimension of $D=80$, but truncate singular values less than 1.0e-8 of the maximum singular value.
In the main text, the ITE results are compared against the ground state obtained via the density matrix renormalization group (DMRG)). 
This should be equivalent to comparing to a long-time ITE ground state.  The long-time ITE ($\beta=38.352$) ground state reached an energy per site 
of -1.455071, while the DMRG ground-state energy per site is -1.455076. The relative error of the nearest neighbor correlations is on the order of 
$10^{-4}$ to $10^{-3}$, and about $10^{-2}$ for correlations between the middle site and the end sites (a distance of 25 sites). 
The error in fidelity between the two ground states was about $5\times 10^{-4}$.

\noindent \textbf{PEPS}. We carried out finite PEPS \cite{nishino1996corner,verstraete2004renormalization,
verstraete2006criticality,orus2014practical} imaginary time evolution for the two-dimensional transverse
field Ising model on a lattice
size of $21 \times 31$. The size was chosen to be large enough to see the spread of mutual information in the bulk 
without significant effects from the boundary. The mutual information was calculated
along the long (horizontal) axis in the center of the lattice.
The standard Trotterized
imaginary time evolution scheme for PEPS~\cite{VerstraeteITimeReview} was used with a time step $\Delta \tau = 0.001$,
up to imaginary time $\beta = 6.0$, starting from a random product state. To reduce computational cost from the
large lattice size, the PEPS
was defined in a translationally invariant manner with only 2 independent tensors \cite{VerstraeteIPEPS}
updated via the so-called ``simple update'' procedure \cite{XiangSimpleUpdate}. 
The simple update has been shown to be sufficiently accurate for capturing correlation 
functions (and thus $I(i,j)$) for ground states with relatively short correlation 
lengths (compared to criticality) \cite{LubaschAlgos,LubaschUnifying}. 
We chose a magnetic field value $h=3.5$ which is 
detuned from the critical field ($h \approx 3.044$) but         
still maintains a correlation length long enough to see
interesting behavior.

\noindent \textit{Accuracy:} Even though the simple update procedure was used for the tensor update,
we still needed to contract the $21 \times 31$ PEPS at at every imaginary time step $\beta$ for a range of
correlation functions, amounting to a large number of contractions.
To control the computational cost, we limited  our bond dimension to $D=5$ and used an optimized
contraction scheme \cite{XiangContract}, with maximum allowed bond dimension 
of $\chi = 60$ during the contraction.
Based on converged PEPS ground state correlation functions
with a larger bond dimension of $D=8$, our $D=5$ PEPS yields $I(i,i+r)$ (where $r$ denotes horizontal separation) at large $\beta$ with
a relative error of $\approx 1\%$ for $r=1-4$, $5\%$ or less for $r=5-8$, and $10\%$ or greater for
$r > 8$. At smaller values of $\beta$ ($< 0.5$) the errors up to $r=8$ are much smaller because 
the bond dimension of 5 is able to completely support the smaller correlations (see Fig. 1, main text).
While error analysis on the 2D Heisenberg model \cite{LubaschAlgos} suggests
that errors with respect to $D=\infty$ may be larger, such analysis also confirms that 
a $D=5$ PEPS captures the qualitative behavior of
correlation in the range $r=5-10$ (and beyond).
Aside from the bond dimension error,
the precision of the calculations is governed by $\chi$ and the lattice size. Using the $21 \times 31$
lattice and $\chi = 60$, we were able to converge entries of single-site
density matrices $\rho(i)$ to a precision of $\pm 10^{-6}$ (two site density matrices $\rho(i,j)$
had higher precision). For $\beta = 0.001-0.012$,
the smallest eigenvalue of $\rho(i)$ fell below this precision threshold, leading
to significant noise in $I(i,j)$. Thus, these values of $\beta$ are omitted from Fig. 1 (main text) 
and the smallest reported values of $I$ are $10^{-6}$, although with more precision we expect $I \to 0$ as $r \to \infty$.
 
Finally, the energy and fidelity errors were computed with respect to the PEPS
ground state of the same bond dimension at $\beta = 10.0$ (10000 time steps).
The convergence of the quantities shown in Fig. 1 (main text) thus isolates the convergence of the
imaginary time evolution, and does not include effects of other errors that 
may result from deficiencies in the wavefunction Ansatz.

\SECONDREVISION{
\subsection{Comparison to classical algorithms}

In the main text, we noted that QITE provided an exponential speedup per iteration over the direct classical implementation of
imaginary time evolution algorithm, given bounded correlation length $C$ during the evolution.
We now compare to some other possible classical algorithms.

We first note that a finite correlation length $C_0$ in the ground-state does not itself imply an efficient classical strategy.
For example, a simple heuristic is to solve the problem locally, e.g. to truncate the problem size at correlation length $C_0$ of the ground-state and solve by exact diagonalization, which can be done in time $\exp(O(C_0d))$ in $d$ spatial dimensions. But this will not generally converge to the correct ground-state in a frustrated Hamiltonian, as this would efficiently solve NP-hard classical satisfiability problems even though these have $C_0 = 0$; physical examples include glassy models. 

Similarly, as QITE defines a quantum circuit for the imaginary time evolution, we might attempt to use it for a faster classical simulation. If we are only interested in local observables, we can apply the circuit in the Heisenberg picture in a classical emulation. However, this gives an extra exponential dependence on the number of previous time-steps: after the unitaries associated to $(e^{-\Delta \tau \hat{h}[1]} e^{-\Delta \tau \hat{h}[2]} \dots)^l$ have been applied, the cost of applying the next unitary scales as $\exp(O(lD))$, with $D$ the domain size of the unitaries, instead of $\exp(O(D))$ in QITE.

Alternatively, if $| \Psi \rangle$ is represented by a tensor network in a classical simulation, then $e^{-\Delta \tau \hat{h}[l]} | \Psi \rangle$ can be represented as a classical tensor network with increased bond dimension \cite{VidalTEBD,Schollwock_Annals_2011}. However, the bond dimension will scale as $\exp(O(lD))$. Further, apart from the extra exponential dependence on $l$, another potential drawback in the tensor network approach is that we cannot guarantee contracting the resulting classical tensor network for an observable is efficient; it is a \#P-hard problem in the worst case in 2D (and even in the average case for Gaussian distributed tensors) \cite{VidalTEBD,Schollwock_Annals_2011,Schuch_PRL_2007,PEPShard2018}.
}

\subsection{Simulation models}

We here define, and give some background on, the models used in the QITE and QLanczos simulations.

\subsubsection{1 qubit field model}

\begin{align}
  \hat{H} = \alpha \hat{X} + \beta \hat{Z}
\end{align}
This Hamiltonian has previously been used as a model for quantum simulations on physical devices in Ref.~\cite{lamm2018simulation}.We used $\alpha= \frac{1}{\sqrt{2}}$ and $\beta = \frac{1}{\sqrt{2}}$. In simulations with this Hamiltonian, the qubit is assumed to be initialized in the $Z$ basis.

\subsubsection{1D Heisenberg and transverse field Ising model}

The 1D short-range Heisenberg Hamiltonian is defined as
\begin{align}
  \hat{H} =\sum_{\langle ij\rangle} \hat{\mathbf{S}}_i \cdot \hat{\mathbf{S}}_j \quad,
\end{align}
\REVISION{
the 1D short-range Heisenberg Hamiltonian in the presence of a field as
\begin{align}
  \hat{H} = J\sum_{\langle ij\rangle} \hat{\mathbf{S}}_i \cdot \hat{\mathbf{S}}_j + B\sum_{i} Z_{i}\quad,
\end{align}
}
the 1D long-range Heisenberg Hamiltonian as
\begin{align}
\hat{H} =\sum_{i \neq j} \frac{1}{|i-j|+1} \, \hat{\mathbf{S}}_i \cdot \hat{\mathbf{S}}_j \quad,
\end{align}
and the 1D AFM transverse-field Ising Hamiltonian as
\begin{align}
  \hat{H} = J\sum_{\langle ij\rangle} \oper{Z}_i  \oper{Z}_j + h\sum_i \oper{X}_i \quad .
\end{align}

\subsubsection{1D Hubbard model}

The 1D Hubbard Hamiltonian is defined as
\begin{align}
  \hat{H} = - \sum_{\langle ij \rangle \sigma} \oper{a}^\dag_{i\sigma} \oper{a}_{j\sigma} + U \sum_i \hat{n}_{i\uparrow} \hat{n}_{i\downarrow}
\end{align}
where $\hat{n}_{i \sigma} = a^\dag_{i\sigma} a_{i\sigma}$, $\sigma \in \{ \uparrow,\downarrow\}$, and $\langle \cdot \rangle$ 
denotes summation over nearest-neighbors, here with open-boundary conditions. We label the $n$
lattice sites with an index $i=0 \dots n-1$, and the $2n-1$ basis functions as $|\varphi_0 \rangle = |0 \uparrow \rangle$,
$|\varphi_1 \rangle = |0 \downarrow \rangle$, $|\varphi_2 \rangle = |1 \uparrow \rangle$,
$|\varphi_3 \rangle = |1 \downarrow \rangle$ $\dots$.
Under Jordan-Wigner transformation, recalling that
\begin{align}
\begin{split}
&\hat{n}_{p} = \frac{1-\oper{Z}_p}{2} \quad, \\
&\hat{a}^\dag_p \hat{a}_q + \hat{a}^\dag_q \hat{a}_p =  \frac{ \oper{X}_p \oper{X}_q \prod_{k=q+1}^{p-1} \oper{Z}_k \left( 1- \oper{Z}_p \oper{Z}_q \right)}{2} \quad,
\end{split}
\end{align}
with $p=0 \dots 2n-2$ and $q<p$, the Hamiltonian takes the form
\begin{align}
\begin{split}
\hat{H} &= - \sum_p \frac{ \oper{X}_{p} \oper{X}_{p+2} \oper{Z}_{p+1} \left( 1- \oper{Z}_{p} \oper{Z}_{p+2} \right)}{2} \\
&+ U \sum_{p \, \mathrm{even}} \frac{(1-\oper{Z}_{2i}) (1-\oper{Z}_{2i+1})}{4} + \mu \sum_p \frac{(1-\oper{Z}_p)}{2}
\end{split}
\end{align}

\subsubsection{H$_2$ molecule minimal basis model}

We use the hydrogen molecule minimal basis model at the STO-6G level of theory. This is a common minimal model of hydrogen chains 
\cite{hachmann2006multireference,Motta_PRX_2017} and has previously been studied in quantum simulations, for example 
in~\cite{OMalley2015}. Given a molecular 
geometry (H-H distance $R$) we perform a restricted Hartree-Fock calculation and express the second-quantized Hamiltonian
in the orthonormal basis of RHF molecular orbitals as~\cite{szaboostlund} 
\begin{equation}
\label{eq:H2}
\hat{H} = H_0 + \sum_{pq} h_{pq} \hat{a}^\dag_p \hat{a}_q + \frac{1}{2} \sum_{prqs} v_{prqs}
\hat{a}^\dag_p \hat{a}^\dag_q \hat{a}_s \hat{a}_r
\end{equation}
where $a^\dag$, $a$ are fermionic creation and annihilation operators for the molecular orbitals.
The Hamiltonian \eqref{eq:H2} is then encoded by a Bravyi-Kitaev transformation into the 2-qubit operator 
\begin{equation}
\hat{H} = g_0 + g_1 \oper{Z}_1 + g_2 \oper{Z}_2 + g_3 \oper{Z}_1 \oper{Z}_2 + g_4 \oper{X}_1 \oper{X}_2 + g_5 \oper{Y}_1 \oper{Y}_2 \;,
\end{equation}
with coefficients $g_i$ given in Table I of \cite{OMalley2015}.
 
\subsubsection{MAXCUT Hamiltonian}

The MAXCUT Hamiltonian encodes the solution of the MAXCUT problem. 
Given a graph $\Gamma = (V,E)$, where $V$
is a set of vertices and $E \subseteq V \times V$ is a set of links between vertices in $V$, a cut of $\Gamma$ is a subset 
$S \subseteq V$ of $V$. The MAXCUT problem consists in finding a cut $S$ that maximizes the number of edges between $S$ and $S^c$ (the complement of $S$).
We denote the number of links in a given cut $S$ as $C(S)$.

In Figure 3 of the main text, we consider a graph $\Gamma$ with vertices and links
\begin{equation}
\begin{split}
V &= \{0,1,2,3,4,5\} \;, \\
E &= \{ (0,3), (1,4), (2,3), (2,4), (2,5), (4,5) \} \;, \\
\end{split}
\end{equation}
respectively. It is easy to verify that $S = \{ 0, 2, 4 \}$, $\{ 0,1,2 \}$, $\{3,4\}$ and their complements $S^c$ 
are solutions of the MAXCUT problem, with weight $C_{max} = 5$.

The MAXCUT problem can be formulated as a Hamiltonian ground-state problem, by (i) associating a qubit to every vertex in $V$, (ii) associating to every
partition $S =$ an element of the computational basis (here assumed to be in the $z$ direction) of the form $| z_0 \dots z_{n-1} \rangle$, where $z_i = 1$ if 
$i \in S$ and $z_i = 0$ if $i \in S^c$, and finding the minimal (most negative) eigenvalue of the $2$-local Hamiltonian 
\begin{equation}
\hat{C} = -\sum_{(ij) \in E} \frac{1 - \oper{Z}_i \oper{Z}_j }{2} \quad .
\end{equation} 
The spectrum of $\hat{C}$ is a subset of numbers $C \in \{ 0,1 \dots |E| \}$.

In the present work, we initialize the qubits in the state $|\Phi \rangle = 
|+\rangle^{\otimes n}$, where $|+\rangle = \frac{|0\rangle + |1\rangle}{\sqrt{2}}$, and evolve $\Phi$ in imaginary time.
Measuring the evolved state at time $\beta$ $|\Phi(\beta) \rangle$ will collapse it onto an element $|z_0 \dots z_{n-1} \rangle$ of the
computational basis, which is also an eigenfunction of $\hat{C}$ with eigenvalue $C$. 
In Figure 3 in the main text, we illustrate the probability $P(|C|=C_{max})$ that such measurements yield a MAXCUT solution. Note that,
even in the presence of oscillations (with the smallest domain size $D=2$) this probability remains above $60 \%$.

\subsection{Numerical simulation details}

\subsubsection{QITE stabilization}

Sampling noise in the expectation values of the Pauli operators can affect the solution to Eq.~\eqref{eq:lineareq} 
that sometimes lead to numerical instabilities. We regularize $\mathbf{S}+\mathbf{S}^T$ against such statistical 
errors by adding a small $\delta$ to its diagonal. To generate the data presented in Figures 2 and 4 of the main 
text, we used $\delta=0.01$ for 1-qubit calculations and $\delta=0.1$ for 2-qubit calculations.

\subsubsection{QLanczos stabilization}

In quantum Lanczos, we generate a set of wavefunctions for different imaginary-time projections of 
an initial state $| \Psi \rangle$, using QITE as a subroutine. The normalized states are
\begin{equation}
| \Phi_l \rangle = \frac{ e^{- l \Delta \tau \hat{H} } | \Psi_T \rangle }{\| e^{- l \Delta \tau \hat{H} } \Psi_T \|} 
\equiv n_l \, e^{- l \Delta \tau \hat{H} } | \Psi_T \rangle \quad 0 \leq l < L_\text{max} \quad .
\end{equation}
where $n_l$ is the normalization constant.
For the exact imaginary-time evolution and $l$, $l^\prime$ both even (or odd) the matrix elements
\begin{equation}
S_{l,l^\prime} = \langle \Phi_l | \Phi_{l^\prime} \rangle 
\quad,\quad
H_{l,l^\prime} = \langle \Phi_l | \hat{H} | \Phi_{l^\prime} \rangle 
\end{equation} 
can be computed in terms of expectation values (i.e. experimentally accessible quantities) only. Indeed, defining
$2r = l+l^\prime$, we have 
\begin{equation}
S_{l,l^\prime} = n_l n_{l^\prime} \, \langle \Psi_T | e^{- l \Delta \tau \hat{H} } e^{- l^\prime \Delta \tau \hat{H} } | \Psi_T \rangle
= \frac{n_l n_{l^\prime}}{n_{r}^2} \quad ,
\end{equation} 
and similarly
\begin{equation}
\begin{split}
H_{l,l^\prime} &= n_l n_{l^\prime} \, \langle \Psi_T | e^{- l \Delta \tau \hat{H} } \hat{H} e^{- l^\prime \Delta \tau \hat{H} } | \Psi_T \rangle
= \\
&= \frac{n_l n_{l^\prime}}{n_{r}^2} \, \langle \Phi_r | \hat{H} | \Phi_r \rangle = S_{l,l^\prime} \, \langle \Phi_r | \hat{H} | \Phi_r \rangle \quad .
\end{split}
\end{equation}
The quantities $n_r$ can be evaluated recursively, since
\begin{equation}
\begin{split}
\frac{1}{n^2_{r+1}} &= \langle \Psi_T | e^{- (r+1) \Delta \tau \hat{H} } e^{- (r+1) \Delta \tau \hat{H} } | \Psi_T \rangle = \\
&= \frac{ \langle \Phi_r | e^{-2 \Delta \tau \hat{H} } | \Phi_r \rangle }{n_r^2} \quad,
\end{split}
\end{equation}
For inexact time evolution, the quantities $n_r$ and $\langle \Phi_r | \hat{H} | \Phi_r \rangle$ can still be used to 
approximate $S_{l,l^\prime}$, $H_{l,l^\prime}$.

Given these matrices, we then solve the generalized  eigenvalue equation $\mathbf{H}\mathbf{x} = E \mathbf{S}\mathbf{x}$ to find an approximation
to the ground-state $| \Phi' \rangle = \sum_l x_l | \Phi_l \rangle$ for the ground state of $\oper{H}$. This eigenvalue equation
can be numerically ill-conditioned, as $S$ can contain small and negative eigenvalues for several reasons (i)
as $m$ increases the vectors $|\Phi_l \rangle$ become linearly dependent; (ii) simulations have finite
precision and noise; (iii) $S$, $H$ are computed approximately when inexact time evolution is performed.

To regularize the problem, out of the set of time-evolved states we extract a better-behaved sequence as follows
(i) start from $|\Phi_\text{last}\rangle = |\Phi_0\rangle$ (ii) add the next $|\Phi_l\rangle$ in the set
of time-evolved states s.t. $|\langle \Phi_l | \Phi_\text{last}\rangle| < s$, where $s$
is a regularization parameter $0<s<1$ (iii) repeat, setting the $|\Phi_\text{last}\rangle=\Phi_l$ (obtained from (ii)), until
the desired number of vectors is reached.
We then solve the generalized eigenvalue equation $\tilde{\mathbf{H}}\mathbf{x} = E \tilde{\mathbf{S}}\mathbf{x}$ 
spanned by this regularized sequence, removing any eigenvalues of $\tilde{\mathbf{S}}$ less than a threshold $\epsilon$. 
The exact emulated QLanczos calculations reported in the main text were stabilized with this algorithm (the source
of error here is primarily (iii)) using stabilization parameter $s=0.95$ and $\epsilon = 10^{-14}$.
The stabilization parameters used in the 
QVM and QPU QLanczos calculations were $s=0.75$ and $\epsilon = 10^{-2}$ (the main source
of error in the simulations was (ii)). Note that the stabilization procedure is unlikely to fix all possible numerical instabilities, but was sufficient
for all models and calculations performed in this work.

\subsection{METTS algorithm}

The METTS (minimally entangled typical thermal state) algorithm~\cite{Stoudenmire2010,White2009} is a sampling method to
calculate thermal properties
based on imaginary time evolution. Consider the thermal average of 
an observable $\hat{O}$
\begin{equation}
\langle \hat{O} \rangle = \frac{1}{Z}\mathrm{Tr}[ e^{-\beta \hat{H}} \hat{O} ] 
=  \frac{1}{Z}\sum_{i} \langle i | e^{- \frac{\beta}{2} \oper{H} } \, \hat{O} \, e^{- \frac{\beta}{2} \oper{H} } | i \rangle \;,
\end{equation}
where $\{|i\rangle\}$ is an orthonormal basis set, and $Z$ is the partition 
function. Defining $|\phi_i\rangle = P_i^{-1/2} e^{- \frac{\beta}{2} \oper{H} } | i \rangle$ with
$P_i = \langle i | e^{-\beta \oper{H}} | i \rangle$,
we obtain
\begin{equation}\label{eq:thermal_sum}
\langle \hat{O}\rangle = \frac{1}{Z} \sum_i P_i \langle \phi_i| \hat{O} |\phi_i\rangle \;.
\end{equation}
The summation in Eq. (\ref{eq:thermal_sum}) can be estimated by sampling 
 $|\phi_i\rangle$ with probability $P_i/Z$, and summing 
the sampled $\langle \phi_i|\hat{O}|\phi_i\rangle$. 

In standard Metropolis sampling for thermal states, one starts from $| \phi_i \rangle$ and obtains the next state
$| \phi_j \rangle$ from randomly proposing and accepting based an
acceptance probability. However, rejecting and resetting  in the quantum analog of Metropolis~\cite{Temme2011} is complicated to
implement on a quantum computer, requiring deep circuits.
The METTS algorithm provides an alternative way to sample 
$| \phi_i \rangle$ distributed with probability $P_i/Z$ without this complicated procedure. 
The algorithm is as follows:

\begin{enumerate}
\item Choose a classical product state (PS) $| i \rangle$.
\item Compute $| \phi_i \rangle = P_i^{-1/2} e^{- \frac{\beta}{2} \oper{H} } | i \rangle$ and 
calculate observables of interest.
\item Collapse the state $| \phi_i \rangle$ to a new PS $|i'\rangle$ with probability
$p(i \rightarrow i') = |\langle i' | \phi_i \rangle|^2$ and repeat Step 2.
\end{enumerate}
In the above algorithm, $|\phi_i\rangle$ is named a minimally entangled typical
thermal state (METTS). 
One can easily show that the set of METTS sampled following the above 
procedure has the correct Gibbs distribution~\cite{Stoudenmire2010}.
Generally, $\{|i\rangle\}$ can be any orthonormal basis. 
For convenience when implementing METTS on a quantum computer, 
$\{|i\rangle\}$ are chosen to be product states. 
On a quantum emulator or a quantum computer, the METTS algorithm is carried out as following:

\begin{enumerate}
\item Prepare a product state $|i\rangle$.
\item Imaginary time evolve $|i\rangle$ with the QITE algorithm to 
$|\phi_i\rangle = P_i^{-1/2} e^{- \frac{\beta}{2} \oper{H} } | i \rangle$, and measure 
the desired observables.
\item Collapse $|\phi_i\rangle$ to another product state by measurement.
\end{enumerate}
In practice, to avoid  long statistical correlations between samples, we used
the strategy of collapsing METTS onto alternating basis sets~\cite{Stoudenmire2010}.  
For instance, for the odd METTS steps, $|\phi_i\rangle$ is collapsed
onto the $X$-basis (assuming a $Z$ computational basis, tensor products of $|+\rangle$ and $|-\rangle$), and for
the even METTS steps, $|\phi_i\rangle$ is collapsed onto the $Z$-basis
(tensor products of $|0\rangle$ and $|1\rangle$). The statistical error is then estimated by block analysis~\cite{Flyvbjerg1989}.

\subsection{Implementation on emulator and quantum processor}

We used pyQuil, an open source Python library, to express quantum circuits that interface with both Rigetti's quantum virtual machine (QVM) and
the Aspen-1 quantum processing units (QPUs).

pyQuil provides a way to include noise models in the QVM simulations. Readout error can be included in a
high-level API provided in the package and is characterized by $p_{00}$ (the probability of reading $|0\rangle$ given that the qubit
is in state $|0\rangle$) and $p_{11}$ (the probability of reading $|1\rangle$ given that the qubit is in state $|1\rangle$). Readout errors
can be mitigated by estimating the relevant probabilities and correcting the estimated expectation values. We do so by using a high level
API present in pyQuil. A general noise model can also be applied to a gate in the circuit by applying the appropriate Kraus maps.
Included in the package is
a high level API that applies the same decoherence error attributed to energy relaxation and dephasing to every gate in the circuit.
This error channel is characterized by the relaxation time $T_{1}$ and coherence time $T_{2}$. We also include in our emulation
our own high-level API that applies the same depolarizing noise channel to every single gate by using the appropriate Kraus maps.
The depolarizing noise is characterized by $p_{1}$, the depolarizing probability for single-qubit gates and $p_{2}$, the
depolarizing probability for two-qubit gates. We do not include all sources of error in our emulation. We applied the same depolarizing and dephasing channels to each gate operation for all qubits, when in reality, they can vary from qubit to qubit. In addition, noise due to cross-talk between qubits cannot
be modeled using the QVM and is another source of discrepancy between the QVM and QPU results.

We investigate the influence of noise on the 2-qubit results obtained via the QVM using different noise parameters; 
\begin{itemize}
\REVISION{
\item Noise model 1:

\begin{tabular}{lll}
$p_{00} = 0.95$ & & $p_{11} = 0.95$ \\
$T_{1} = 10.5 \, \mu s$ & & $T_{2} = 14.0 \, \mu s$ \\
$p_{1} = 0.001$ & & $p_{2} = 0.01$ \\
\end{tabular}

\item Noise model 2: 

\begin{tabular}{lll}
$p_{00} = 0.99$ & & $p_{11} = 0.99$ \\
$T_{1} = 10.5 \, \mu s$ & & $T_{2} = 14.0 \, \mu s$ \\
$p_{1} = 0.001$ & & $p_{2} = 0.01$ \\
\end{tabular}

\item Noise model 3: 

\begin{tabular}{lll}
$p_{00} = 0.99$ & & $p_{11} = 0.99$ \\
$T_{1} = 20.0 \, \mu s$ & & $T_{2} = 40.0 \, \mu s$ \\
$p_{1} = 0.0001$ & & $p_{2} = 0.001$ \\
\end{tabular}
}
\end{itemize}  
Noise model 1 reflects realistic parameters that characterize the Aspen-1 QPUs we run our calculations on; 
$p_{00}$, $p_{11}$, $T_{1}$, and $T_{2}$ are reported values whereas  $p_{1}$ and $p_{2}$ are values typically 
used to benchmark error mitigation algorithms~\cite{Temme2017}. We repeated 10 calculations for each noise model 
and note there is practically no variation from run to run. Fig.~\ref{fig:SI_1}(a) shows that reducing the readout error 
does not greatly affect the converged ground state energy after readout error mitigation has been performed. 
However, reducing the other sources of error does improve the converged energy. 
Note that sufficient measurement samples are used such that the sampling variance is smaller than that due to noise.

We also ran 2-qubit simulations on different pairs of qubits on Aspen-1, with Q1 consisting of qubits 14, 15 and Q2 consisting of qubits 0,1. 
These two pairs are reported to have different noise characteristics,
\begin{itemize}
\REVISION{
\item Q1: 

\begin{tabular}{lll}
$p_{00} = 0.95$ & & $p_{11} = 0.95$ \\
$T_{1} = 10.5 \, \mu s$ & & $T_{2} = 14.0 \, \mu s$ \\ 
\end{tabular}

\item Q2: 

\begin{tabular}{lll}
$p_{00} = 0.90$ & & $p_{11} = 0.90$ \\
$T_{1} = 6.5 \, \mu s$ & & $T_{2} = 8.0 \, \mu s$ \\ 
\end{tabular}

}
\end{itemize}
Based on this, we expect simulations on Q2 to be worse.
Note that in contrast to our QVM calculations, the results from the actual devices varied from run to run. 
Thus, we present the mean and standard deviation for 10 different runs on each pair. 
(Similarly, sufficient samples are taken when running the QVM such that the sampling variance is smaller than that due to noise).
Fig.~\ref{fig:SI_1}(b) indeed demonstrates that Q2 provides a less faithful implementation of the quantum algorithm.

\begin{figure}[h!]
\includegraphics[width=\columnwidth]{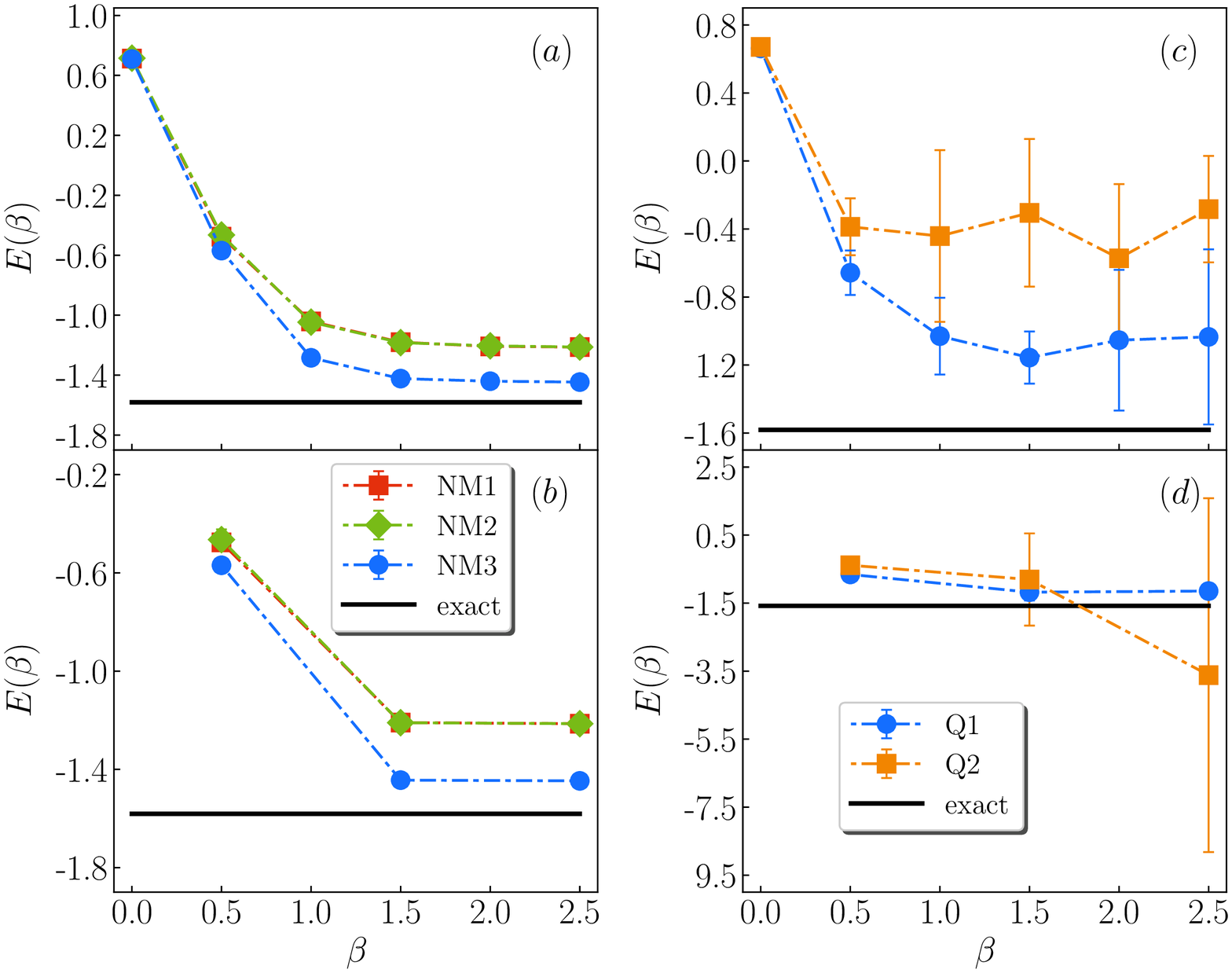} 
\caption{Comparison of energies obtained using different noise models(NM) for (a) QITE and (b) QLanczos. Comparison of energies obtained using different pair of qubits for (c) QITE and (d) QLanczos. The performance of QITE and QLanczos improves as noise is reduced, indicating the potential of the algorithms. 
}
\label{fig:SI_1}
\end{figure}

\subsubsection{Parameters used in QVM and QPUs simulations}

In this section, we include the parameters used in our QPU and QVM simulations.
Note that all noisy QVM simulations (unless stated otherwise in the text) were performed with noise parameters from noise model 1. 
We also indicate the number of samples used during measurements for each Pauli operator.

\begin{table}[h!]
\label{tab:table1}
  \begin{center}
    \caption{QPUs: 1-qubit QITE and QLanczos.}
    \begin{tabular}{l|c|c|c|r} 
      \textbf{Trotter stepsize} & \textbf{nSamples} & \textbf{$\delta$} & \textbf{s} & \textbf{$\epsilon$}\\
      \hline
      0.2 & 100000 & 0.01 & 0.75 & $10^{-2}$\\
    \end{tabular}
  \end{center}
\end{table}

\begin{table}[h!]
  \begin{center}
    \caption{QPUs: 2-qubit QITE and QLanczos.}
    \begin{tabular}{l|c|c|c|r} 
      \textbf{Trotter stepsize} & \textbf{nSamples} & \textbf{$\delta$} & \textbf{s} & \textbf{$\epsilon$}\\
      \hline
      0.5 & 100000 & 0.1 & 0.75 & $10^{-2}$\\
    \end{tabular}
  \end{center}
\end{table}

\begin{table}[h!]
  \begin{center}
    \caption{QPUs: 1-qubit METTS.}
    \begin{tabular}{l|c|c|c|r} 
      \textbf{$\beta$} & \textbf{Trotter stepsize} & \textbf{nSamples} & \textbf{nMETTs} & \textbf{$\delta$}\\
      \hline
      1.5 & 0.15 & 1500 & 70 & 0.01\\
      2.0 & 0.20 & 1500 & 70 & 0.01\\
      3.0 & 0.30 & 1500 & 70 & 0.01\\
      4.0 & 0.40 & 1500 & 70 & 0.01\\
    \end{tabular}
  \end{center}
\end{table}

\begin{table}[h!]
  \begin{center}
    \caption{QVM: 2-qubit QITE and QLanczos.}
    \begin{tabular}{l|c|c|c|r} 
      \textbf{Trotter stepsize} & \textbf{nSamples} & \textbf{$\delta$} & \textbf{s} & \textbf{$\epsilon$}\\
      \hline
      0.5 & 100000  & 0.1 & 0.75  & $10^{-2}$\\
    \end{tabular}
  \end{center}
\end{table}

\begin{table}[h!]
  \begin{center}
    \caption{QVM: 1-qubit METTS.}
    \begin{tabular}{l|c|c|c|r} 
      \textbf{$\beta$} & \textbf{Trotter stepsize} & \textbf{nSamples} & \textbf{nMETTs} & \textbf{$\delta$}\\
      \hline
      1.0 & 0.10 & 1500 &70 & 0.01\\
      1.5 & 0.15 & 1500 &70 & 0.01\\
      2.0 & 0.20 & 1500 &70 & 0.01\\
      3.0 & 0.30 & 1500 &70 & 0.01\\
      4.0 & 0.40 & 1500 &70 & 0.01\\
    \end{tabular}
  \end{center}
\end{table}

\begin{table}[h!]
 \begin{center}
    \caption{QVM: 2-qubit METTS.}
    \begin{tabular}{l|c|c|c|r} 
      \textbf{$\beta$} & \textbf{Trotter stepsize} & \textbf{nSamples} & \textbf{nMETTs} & \textbf{$\delta$}\\
      \hline
      1.0 & 0.10 & 30000 & 100 & 0.1\\
      1.5 & 0.15 & 30000 & 100 & 0.1\\
      2.0 & 0.20 & 30000 & 100 & 0.1\\
      3.0 & 0.30 & 30000 & 100 & 0.1\\
      4.0 & 0.40 & 30000 & 100 & 0.1\\
    \end{tabular}
  \end{center}
\end{table}

\clearpage

\subsection{Comparison of QITE and VQE}

To address the feasibility of running QITE for larger systems on near-term devices, we compared the total number of
Pauli string measurements needed for both VQE and QITE to obtain the ground state of two different spin models; (a) a 1D Heisenberg chain in a magnetic field with
the parameters $J=B=1$; the 4-site instance of this model was  studied in Ref.~\cite{kandala2017}, and (b) 1D AFM transverse-field Ising model ($J=h=1/\sqrt{2}$). Specifically, we estimated how
many expectation values of Pauli strings would need to be measured to obtain the ground state of a 4-site and 6-site instance. The state under the evolution of QITE or in VQE was said to be converged to the ground state if its energy was within 1$\%$  of the exact ground state energy for the Heisenberg
model, and 1$\%$ or 2$\%$ for the Ising model (the relaxed criterion for the Ising model was chosen so that the VQE optimization could
complete in a reasonable number of steps).
The following section describes how we counted the total number of Pauli string measurements in VQE and QITE.

\subsubsection{Counting Pauli strings in VQE}
To perform the VQE calculations, we used the hardware-efficient variational Ansatz as described in ~\cite{kandala2017}. This consists of first applying rotation unitaries represented by $U^{q,i}(\boldsymbol{\theta}) = Rz_{\theta^{q,i}_{1}}Rx_{\theta^{q,i}_{2}}Rz_{\theta^{q,i}_{3}}$ to all qubits before applying layers of a certain depth $d$; each layer begins by applying $CZ$ gates between nearest-neighbors followed by applying $U^{q,i}(\boldsymbol{\theta}) = Rz_{\theta^{q,i}_{1}}Rx_{\theta^{q,i}_{2}}Rz_{\theta^{q,i}_{3}}$ to all qubits again. Details of the circuit can be seen in Fig.~\ref{fig:vqe_circuit}. 

\begin{figure*}[t!]
\includegraphics[width=0.7\textwidth]{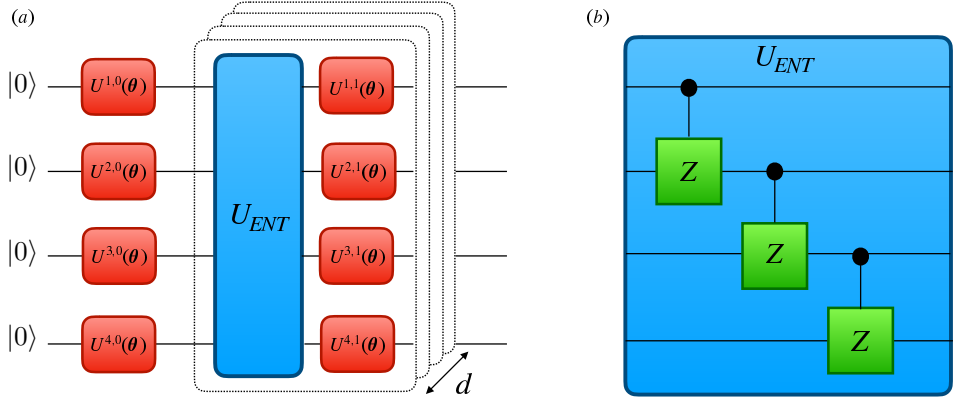}
\caption{\REVISION{(i) VQE Ansatz that is composed of a sequence of interleaved single-qubit rotations $U^{q,i}(\boldsymbol{\theta})$ and entangling operations. (ii) The entangling operations consist of applying $CZ$ gates between nearest neighbours.}}
\label{fig:vqe_circuit}
\end{figure*}

As in ~\cite{kandala2017}, we also used the simultaneous perturbation stochastic approximation (SPSA) algorithm as the optimization protocol. The SPSA algorithm is commonly used because (i) it performs well in the presence of stochastic fluctuations and (ii) it requires only evaluating the objective function twice to update the variational parameters regardless of the number of parameters involved. The performance of the optimizer depends on the hyperparameters $\alpha$ and $\gamma$ as described in ~\cite{kandala2017} and we found that their reported values of $\alpha=0.602$ and $\gamma=0.101$ also gave the best results for us. Numerical evidence of this is provided later on.

Evaluating the objective function involves estimating the expectation value of the Pauli strings that appear in the Hamiltonian. To prevent sampling errors from influencing the comparison,  we evaluated the expectation values exactly. We conducted the VQE calculations using Qiskit, a quantum emulator Python package provided by IBM. The package provides both the SPSA algorithm and a variational Ansatz which we modified to reproduce the exact Ansatz used in~\cite{kandala2017}.

\begin{figure}[t!]
\includegraphics[width=\columnwidth]{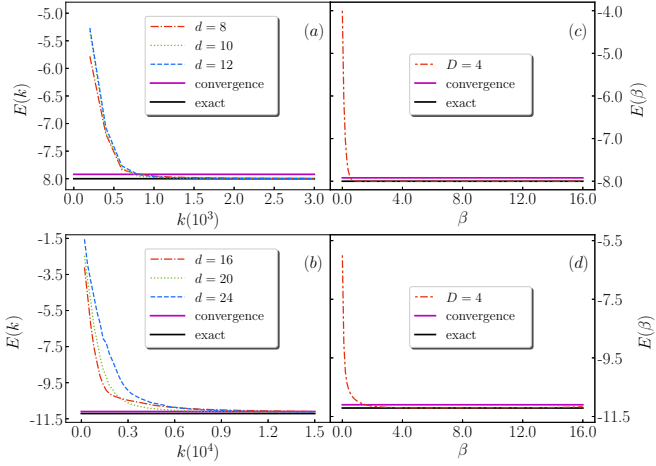} 
\caption{\REVISION{VQE calculations for (a) 4-site and (b) 6-site 1D Heisenberg model. $k$ is the number of optimization steps, and $d$ is the number of layers.
  QITE calculations for (c) 4-site and (d) 6-site for the same model. $D$ is the domain size.}}
\label{fig:calculations}
\end{figure}

To count the number of Pauli strings needed for convergence, we ran VQE using different layer depths and determined the number of iterations $N$ needed for the algorithm to converge to a state with an energy  within a certain percentage (1$\%$ or 2$\%$) of the exact ground state energy. Examples of converged VQE calculations for the 1D Heisenberg model are given in Fig.~\ref{fig:calculations}(a) and (b). In each iteration, the objective function was evaluated twice, and the evaluation of the objective function required measuring the expectation value of the $M$ Pauli strings that appear in the Hamiltonian. Therefore, the total number of
Pauli strings $P_{total}$ is given as
\begin{equation}
P_{total} = 2 \times N \times M
\end{equation}
We note that the results from one VQE trajectory can differ slightly from the next. For our VQE calculations, we always performed 10 trajectories and analyzed our data using the average trajectory. We summarize the VQE parameters that we found gave the
lowest number of total Pauli measurements to converge to the ground state in Table ~\ref{table:vqe}. For the 6-site 1D AFM transverse field Ising model, VQE could not converge to within 1$\%$ and we instead used the 2$\%$ convergence criterion.  We found that for $\alpha = 0.602, \gamma = 0.101$, our simulation results for the 6-site 1D Heisenberg model indicates that using a circuit depth of 20 requires the least number of total Pauli measurements. We also ran some tests to determine what values of $\alpha $ and $\gamma$ gave the best result for the 6-site 1D Heisenberg model; the VQE calculation for the 6-site model conducted using $\alpha=0.602$, $\gamma=0.101$ and a circuit depth of 20 converged within 8400 optimization steps. We ran VQE calculations for different $\alpha$ and $\gamma$ using the same circuit depth of 20 and a total of 9000 optimization steps. The data in table ~\ref{table:hyperparameter} clearly shows that $\alpha=0.602$ and $\gamma=0.101$ produced the best result for us.  

\subsubsection{Counting Pauli strings in QITE}
To implement QITE, we used the second-order Trotter decomposition given by
\begin{align}
  e^{-\beta \oper{H}} &= (e^{-\ts/2 \operSB{h}{1}} \ldots e^{-\ts/2 \operSB{h}{K-1}}e^{-\ts \operSB{h}{K}}  \\ \nonumber
  &\quad e^{-\ts/2 \operSB{h}{K-1}} \ldots e^{-\ts \operSB{h}{1}})^n + \mathcal{O}\left( {\ts}^{2} \right); \ n= \frac{\beta}{\ts}
\end{align}
to carry out the real time evolution. We initialized our state as: (a) $|0101\ldots \rangle$  for the 1D Heisenberg model and (b) maximally-mixed state for the 1D AFM transverse-field Ising model. We converged to the ground state using a time step of $\ts = 0.1$ and a domain size $D$ of 4, as seen in Figs.~\ref{fig:calculations}(c) and (d). To count the number of Pauli strings, we note that a domain size of 4 implies that to evaluate $e^{-\ts/2 \operSB{h}{i}}$ involves measuring $4^{4} = 256$ Pauli strings (without using the real-valued nature of the Hamiltonian). Therefore, with a total number of Trotter steps $T$, the total number of Pauli strings $P_{total}$ is given as
\begin{equation}
P_{total} = (2K-1) \times T \times 256
\end{equation}
We summarize the parameters that we used for QITE to obtain the ground state in table ~\ref{table:qite}. We had no trouble converging our ground state
to arbitrary accuracy using QITE but we used the same convergence criterion as for VQE to facilitate comparison.

\subsubsection{VQE and QITE}
Data from Table ~\ref{table:qite_vqe} suggests that QITE is competitive with VQE with respect to the number of Pauli string measurements.
In fact, for the 6 qubit system, the number of measurements needed in QITE was significantly less than in VQE, due largely
to the SPSA  iterations needed to reach convergence when optimizing the VQE energy. While it is likely that the VQE costs
could be lowered by using a better optimizer, or a better VQE Ansatz, we also note that the counts for QITE can also be reduced by using the methods outlined in the main text and earlier sections that discussed how one can economize measurements in QITE.
The widespread current implementation
of VQE and the observed performance of QITE suggest that it will be practical to implement the QITE protocol for intermediate system sizes
on near-term devices.

\begin{table}[htp!]
 \begin{center}
   \caption{\REVISION{VQE simulation parameters for (a) 1D Heisenberg with applied field and (b) 1D AFM transverse field Ising. Conv. refers to the convergence criterion used. We note for the last case, the VQE optimization could not reach within 1$\%$ of the ground state energy,
       so we set the convergence criterion to 2 $\%$}}
    \label{table:vqe}
    \begin{tabular}{c|c|c|c|c|c|c|r} 
      model & n-site & conv. & $\alpha$ & $\gamma$ & $d$ & $N$ &$P_{total}$\\
      \hline
      a & 4 & 1$\%$ & 0.602 & 0.101 & 8 & 800 & 25,600\\
      a & 6 & 1$\%$ & 0.602 & 0.101 & 20 & 8400 & 403,200\\
      b & 4 & 1$\%$ & 0.602 & 0.101 & 12 & 800 & 12,800\\
      b & 6 & 2$\%$ & 0.602 & 0.101 & 12 & 2890 & 69,360\\
    \end{tabular}
  \end{center}
\end{table}

\begin{table}[htp!]
 \begin{center}
    \caption{\REVISION{Hyperparameters sweep for 6-site 1D Heisenberg model using a circuit depth of 20 for a total of 9000 optimization steps. The step at which the calculation converged is recorded under column $T$. '-' indicates that VQE failed to converge.}}
    \label{table:hyperparameter}
    \begin{tabular}{l|c|r} 
      $\alpha$ & $\gamma$ & $T$\\
      \hline
      0.400 & 0.066 & -\\
      0.400 & 0.101 & -\\
      0.400 & 0.133 & -\\
      0.602 & 0.066 & -\\
      0.602 & 0.101 & 8400\\
      0.602 & 0.133 & 8800\\
      0.800 & 0.066 & -\\
      0.800 & 0.101 & -\\
      0.800 & 0.133 & -\\
    \end{tabular}
  \end{center}
\end{table}

\clearpage
\begin{table}[htp!]
 \begin{center}
    \caption{\REVISION{QITE simulation parameters for (a) 1D Heisenberg with applied field and (b) 1D AFM transverse field Ising. Conv. indicates the convergence criterion used. We used 2$\%$ for the final calculation to facilitate comparison with VQE which failed to converge to within 1$\%$.}}
    \label{table:qite}
    \begin{tabular}{c|c|c|c|c|c|c|r} 
      model & n-site & conv. & $\ts$ & $D$ & $T$  & K & $P_{total}$\\
      \hline
      a & 4 & 1$\%$ & 0.1 & 4 & 7 & 4 & 12,544\\
      a & 6 & 1$\%$ & 0.1 & 4 & 17 & 6 & 47,872\\
      b & 4 & 1$\%$ & 0.2 & 4 & 7 & 4 & 12,544\\
      b & 6 & 2$\%$ & 0.2 & 4 & 8 & 6 & 22,528\\
    \end{tabular}
  \end{center}
\end{table}

\begin{table}[htp!]
 \begin{center}
   \caption{\REVISION{Total Pauli string expectation values in VQE and QITE for (a) 1D Heisenberg with applied field and (b) 1D AFM transverse field Ising. The total number of Pauli strings to be measured in QITE can be further reduced by using only Pauli strings with only an odd number of $\hat{Y}$ operators due to the
       real nature of the Hamiltonian. We show this reduced number in brackets.}}
    \label{table:qite_vqe}
    \begin{tabular}{c|c|r|r} 
      model & n-site & VQE & QITE\\
      \hline
      a & 4 & 25,600 & 12,544(5,880)\\
      a & 6 & 403,200 & 47,872(22,440)\\
      b & 4 & 12,800 & 12,544(5,880)\\
      b & 6 & 69,360 & 22,528(10,560)\\
    \end{tabular}
  \end{center}
\end{table}

\end{document}